\documentclass{article}

\usepackage{amsfonts,amsmath,apxproof,balance,complexity,hyperref,pgfplots,tikz,tikz-3dplot}
\usetikzlibrary{shapes.geometric}
\usetikzlibrary{positioning}

\usepackage[sort,nocompress]{cite}
\usepackage{subcaption}
\usepackage{blindtext}
\usepackage[margin=1in]{geometry}

\theoremstyle{plain}

\newtheorem{lemma}{Lemma}
\newtheorem{open}{Open problem}

\newtheorem{rem}{Remark}
\newtheorem{clm}{Claim}
\newtheorem{obs}{Observation}

\theoremstyle{definition}

\newtheoremrep{thm}{Theorem}
\newtheoremrep{lemma}{Lemma}
\newtheoremrep{rem}{Remark}
\newtheoremrep{clm}{Claim}
\newtheoremrep{obs}{Observation}
\newtheoremrep{theorem}{Theorem}
\newtheoremrep{col}{Corollary}
\newtheoremrep{redu}[theorem]{Reduction}

\def\apxmark{$\!\!${\bf(*)}}

\makeatletter
\newcommand\footnoteref[1]{\protected@xdef\@thefnmark{\ref{#1}}\@footnotemark}
\makeatother

\def\clause#1#2#3#4{
	\begin{scope}[square/.style = {regular polygon, regular polygon sides=4, scale=1.2}]
		
		\node[fill=black, opacity=1, draw=black, square] at (#1,#2) {};
		\node[fill=black, opacity=1, draw=black, square] at (#1 + 0.6, #3) {};
		\node[fill=black, opacity=1, draw=black, square] at (#1 + 1.2,#4) {};
		
		\draw (#1,#2) -- (#1,0);
		\draw (#1 + 0.6, #3) -- (#1+ 0.6,0);
		\draw (#1 + 1.2,#4) -- (#1+ 1.2,0);
	\end{scope}
}

\def\crossing#1#2{
	\begin{scope}[shift={#1}, scale=#2]
		\node (A) at (0,0.7) {}; 
		\node (B) at (0.6,1.5) {}; 
		
		\node (C) at (0.4,-0.1) {}; 
		\node (D) at (0.6,0.3) {}; 
		
		\node (E) at (1.4,0.1) {}; 
		\node (F) at (1.4,0.7) {}; 
		
		\node (G) at (1.2,1.1) {}; 
		\node (H) at (1.4,2) {}; 
		\node (I) at (1.85,0.75) {}; 
		\node (J) at (2.1,1.6) {}; 
		
		\node (K) at (2.2,-0.2) {}; 
		\node (L) at (2.2,0.8) {}; 
		
		\node (M) at (1.3,-1) {}; 
		\node (N) at (1.3,-0.5) {}; 
		
		\node (O) at (2,-1.4) {}; 
		\node (P) at (2,-0.8) {}; 
		
		\node (Q) at (2.7,0) {}; 
		\node (R) at (3.1,0) {}; 
		
		\node (L1) at (-2,1) {};
		\node (L2) at (-1.2,0.6) {};
		
		\node (L3) at (-0.6,0.9) {};
		\node (L4) at (-0.6,0.3) {};
		
		\node (R1) at (3.5,0.3) {};
		\node (R2) at (3.5,-0.3) {};
		
		\node (R3) at (4,0) {};
		\node (R4) at (4,-0.6) {};
		
		\node (R5) at (4.5,0.7) {};
		\node (R6) at (4.5,-0.3) {};
		
		\node (U1) at (1,2.5) {};
		\node (U2) at (1.6,2.5) {};
		
		\node (U3) at (1.4,3) {};
		\node (U4) at (1.4,4) {};
		
		\node (D1) at (1.8,-2.1) {};
		\node (D2) at (2.4,-1.9) {};

		\draw[thick] (A)--(B);
		\draw[thick] (C)--(D);
		\draw[thick] (E)--(F);
		\draw[thick] (G)--(H);
		\draw[thick] (I)--(J);
		\draw[thick] (K)--(L);
		\draw[thick] (M)--(N);
		\draw[thick] (O)--(P);
		\draw[thick] (Q)--(R);
		\draw[thick] (L1)--(L2);
		\draw[thick] (L3)--(L4);
		\draw[thick] (R1)--(R2);
		\draw[thick] (R3)--(R4);
		\draw[thick] (R5)--(R6);
		\draw[thick] (D1)--(D2);
		\draw[thick] (U1)--(U2);
		\draw[thick] (U3)--(U4);
		
	\end{scope}
}

\def\hedgeT#1#2{
	\begin{scope}[shift={#1}]
		\foreach [evaluate={\j = int(mod(\i,2)); \x = int(mod(\i,3));}] \i in {1,...,#2}
		{
			\ifthenelse{\x = 1}
			{
				\node[BLUE] (\i) at (\i - 0.2, 0) {};
				\node[RED] (\i') at (\i + 0.2 , 0.3) {};
			}
			{}
			
			\ifthenelse{\x = 2}
			{
				\node[GREEN] (\i) at (\i - 0.2, 0) {};
				\node[BLUE] (\i') at (\i + 0.2 , 0.3) {};
				
			}
			{}
			
			\ifthenelse{\x = 0}
			{
				\node[RED] (\i) at (\i - 0.2, 0) {};
				\node[GREEN] (\i') at (\i + 0.2 , 0.3) {};
			}
			{}

			\ifthenelse{\j = 0}
			{
				\pgfmathtruncatemacro\k{\i-1};
				\draw[black, thick] (\i)--(\k);
				\draw[black, thick] (\i')--(\k');		
			}
			{}			
		}
		
	\end{scope}
}

\def\hedgeF#1#2{
	\begin{scope}[shift={#1}]
		\foreach [evaluate={\j = int(mod(\i,2)); \x = int(mod(\i,3));}] \i in {1,...,#2}
		{
			\ifthenelse{\x = 1}
			{
				\node[RED] (\i) at (\i - 0.2, 0) {};
				\node[BLUE] (\i') at (\i + 0.2 , 0.3) {};
			}
			{}
			
			\ifthenelse{\x = 2}
			{
				\node[GREEN] (\i) at (\i - 0.2, 0) {};
				\node[RED] (\i') at (\i + 0.2 , 0.3) {};
				
			}
			{}
			
			\ifthenelse{\x = 0}
			{
				\node[BLUE] (\i) at (\i - 0.2, 0) {};
				\node[GREEN] (\i') at (\i + 0.2 , 0.3) {};
			}
			{}

			\ifthenelse{\j = 0}
			{
				\pgfmathtruncatemacro\k{\i-1};
				\draw[black, thick] (\i)--(\k);
				\draw[black, thick] (\i')--(\k');		
			}
			{}			
		}
		
	\end{scope}
}

\def\hedgeL#1#2{
	\begin{scope}[shift={#1}]
		\foreach [evaluate={\j = int(mod(\i,2)); \x = int(mod(\i,3));}] \i in {1,...,#2}
		{
			\ifthenelse{\x = 1}
			{
				\node[BLUE] (\i) at (\i + 0.2, 0) {};
				\node[RED] (\i') at (\i - 0.2 , 0.3) {};
			}
			{}
			
			\ifthenelse{\x = 2}
			{
				\node[RED] (\i) at (\i + 0.2, 0) {};
				\node[GREEN] (\i') at (\i - 0.2 , 0.3) {};
				
			}
			{}
			
			\ifthenelse{\x = 0}
			{
				\node[GREEN] (\i) at (\i + 0.2, 0) {};
				\node[BLUE] (\i') at (\i - 0.2 , 0.3) {};
			}
			{}

			\ifthenelse{\j = 0}
			{
				\pgfmathtruncatemacro\k{\i-1};
				\draw[black, thick] (\i)--(\k);
				\draw[black, thick] (\i')--(\k');		
			}
			{}			
		}
		
	\end{scope}
}

\def\nae#1{
	\begin{scope}[shift={#1}]
		\node (A) at (-1.4,0.4) {}; 
		\node (B) at (-1.4,-0.4) {}; 
		
		\node (C) at (-1,0) {}; 
		\node (D) at (-0.4,0) {}; 
		
		\node (E) at (1.4,0.4) {}; 
		\node (F) at (1.4,-0.4) {}; 
		
		\node (G) at (1,0) {}; 
		\node (H) at (0.4,0) {}; 
		
		\node (I) at (-0.4,1.4) {}; 
		\node (J) at (0.4,1.4) {}; 
		
		\node (K) at (0,1) {}; 
		\node (L) at (0,0.4) {}; 
		
		\node (L1) at (-2.7,0.4) {};
		\node (L2) at (-2.7,-0.4) {};
		
		\node (L3) at (-2,0.8) {};
		\node (L4) at (-2,0) {};
		
		\node (R1) at (2.6,0.4) {};
		\node (R2) at (2.6,-0.4) {};
		
		\node (R3) at (2,0.8) {};
		\node (R4) at (2,0) {};
		
		\node (U1) at (-0.4,2.5) {};
		\node (U2) at (0.4,2.5) {};
		
		\node (U3) at (0,1.9) {};
		\node (U4) at (0.8,1.9) {};
		
		\draw[thick] (A)--(B);
		\draw[thick] (C)--(D);
		\draw[thick] (E)--(F);
		\draw[thick] (G)--(H);
		\draw[thick] (I)--(J);
		\draw[thick] (K)--(L);
		\draw[thick] (L1)--(L2);
		\draw[thick] (L3)--(L4);
		\draw[thick] (R1)--(R2);
		\draw[thick] (R3)--(R4);
		\draw[thick] (U1)--(U2);
		\draw[thick] (U3)--(U4);

	\end{scope}
}

\def\vedgeT#1#2{
	\begin{scope}[shift={#1}]
		\foreach [evaluate={\j = int(mod(\i,2)); \x = int(mod(\i,3));}] \i in {1,...,#2}
		{
			\ifthenelse{\x = 1}
			{
				\node[GREEN] (\i) at (0, -\i - 0.3) {};
				\node[BLUE] (\i') at (0.25, -\i + 0.3) {};
			}
			{}
			
			\ifthenelse{\x = 2}
			{
				\node[BLUE] (\i) at (0, -\i - 0.3) {};
				\node[RED] (\i') at (0.25, -\i + 0.3) {};
			}
			{}
			
			\ifthenelse{\x = 0}
			{
				\node[RED] (\i) at (0, -\i - 0.3) {};
				\node[GREEN] (\i') at (0.25, -\i + 0.3) {};
			}
			{}

			\ifthenelse{\j = 0}
			{
				\pgfmathtruncatemacro\k{\i-1};
				\draw[black, thick] (\i)--(\k);
				\draw[black, thick] (\i')--(\k');		
			}
			{}			
		}
		
	\end{scope}
}

\def\vedgeCT#1#2{
	\begin{scope}[shift={#1}]
		\node[GREEN] (A) at (-0.3, 0) {};
		\node[RED] (B) at (0, -0.4) {};
		\draw[black, thick] (A)--(B);
		\foreach [evaluate={\j = int(mod(\i,2)); \x = int(mod(\i,3));}] \i in {1,...,#2}
		{
			\ifthenelse{\x = 1}
			{
				\node[BLUE] (\i) at (0, -\i - 0.3) {};
				\node[GREEN] (\i') at (0.25, -\i + 0.3) {};
			}
			{}
			
			\ifthenelse{\x = 2}
			{
				\node[GREEN] (\i) at (0, -\i - 0.3) {};
				\node[RED] (\i') at (0.25, -\i + 0.3) {};
			}
			{}
			
			\ifthenelse{\x = 0}
			{
				\node[RED] (\i) at (0, -\i - 0.3) {};
				\node[BLUE] (\i') at (0.25, -\i + 0.3) {};
			}
			{}
			
			\ifthenelse{\j = 0}
			{
				\pgfmathtruncatemacro\k{\i-1};
				\draw[black, thick] (\i)--(\k);
				\draw[black, thick] (\i')--(\k');		
			}
			{}			
		}
	\end{scope}
}

\def\vedgeCF#1#2{
	\begin{scope}[shift={#1}]
		\node[GREEN] (A) at (-0.3, 0) {};
		\node[BLUE] (B) at (0, -0.4) {};
		\draw[black, thick] (A)--(B);
		\foreach [evaluate={\j = int(mod(\i,2)); \x = int(mod(\i,3));}] \i in {1,...,#2}
		{
			\ifthenelse{\x = 1}
			{
				\node[RED] (\i) at (0, -\i - 0.3) {};
				\node[GREEN] (\i') at (0.25, -\i + 0.3) {};
			}
			{}
			
			\ifthenelse{\x = 2}
			{
				\node[GREEN] (\i) at (0, -\i - 0.3) {};
				\node[BLUE] (\i') at (0.25, -\i + 0.3) {};
			}
			{}
			
			\ifthenelse{\x = 0}
			{
				\node[BLUE] (\i) at (0, -\i - 0.3) {};
				\node[RED] (\i') at (0.25, -\i + 0.3) {};
			}
			{}
			
			\ifthenelse{\j = 0}
			{
				\pgfmathtruncatemacro\k{\i-1};
				\draw[black, thick] (\i)--(\k);
				\draw[black, thick] (\i')--(\k');		
			}
			{}			
		}
	\end{scope}
}

\def\vedgeL#1#2{
	\begin{scope}[shift={#1}]
		\foreach [evaluate={\j = int(mod(\i,2)); \x = int(mod(\i,3));}] \i in {1,...,#2}
		{
			
			\ifthenelse{\x = 1}
			{
				\node[RED] (\i) at (0, -\i + 0.3) {};
				\node[GREEN] (\i') at (0.25, -\i - 0.3) {};
			}
			{}
			
			\ifthenelse{\x = 2}
			{
				\node[BLUE] (\i) at (0, -\i + 0.3) {};
				\node[RED] (\i') at (0.25, -\i - 0.3) {};
			}
			{}
			\ifthenelse{\x = 0}
			{
				
				\node[GREEN] (\i) at (0, -\i + 0.3) {};
				\node[BLUE] (\i') at (0.25, -\i - 0.3) {};
			}
			{}
			
			\ifthenelse{\j = 0}
			{
				\pgfmathtruncatemacro\k{\i-1};
				\draw[black, thick] (\i)--(\k);
				\draw[black, thick] (\i')--(\k');		
			}
			{}			
		}
	\end{scope}
}

\makeatletter
\def\grd@save@target#1{%
	\def\grd@target{#1}}
\def\grd@save@start#1{%
	\def\grd@start{#1}}
\tikzset{
	grid with coordinates/.style={
		to path={%
			\pgfextra{%
				\edef\grd@@target{(\tikztotarget)}%
				\tikz@scan@one@point\grd@save@target\grd@@target\relax
				\edef\grd@@start{(\tikztostart)}%
				\tikz@scan@one@point\grd@save@start\grd@@start\relax
				\draw[minor help lines] (\tikztostart) grid (\tikztotarget);
				\draw[major help lines] (\tikztostart) grid (\tikztotarget);
				\grd@start
				\pgfmathsetmacro{\grd@xa}{\the\pgf@x/1cm}
				\pgfmathsetmacro{\grd@ya}{\the\pgf@y/1cm}
				\grd@target
				\pgfmathsetmacro{\grd@xb}{\the\pgf@x/1cm}
				\pgfmathsetmacro{\grd@yb}{\the\pgf@y/1cm}
				\pgfmathsetmacro{\grd@xc}{\grd@xa + \pgfkeysvalueof{/tikz/grid with coordinates/major step}}
				\pgfmathsetmacro{\grd@yc}{\grd@ya + \pgfkeysvalueof{/tikz/grid with coordinates/major step}}
				\foreach \x in {\grd@xa,...,\grd@xb}
				\node[anchor=north] at (\x,\grd@ya) {\tiny \pgfmathprintnumber{\x}};
				\foreach \y in {\grd@ya,...,\grd@yb}
				\node[anchor=east] at (\grd@xa,\y) {\tiny \pgfmathprintnumber{\y}};
			}
		}
	},
	minor help lines/.style={
		help lines,
		step=\pgfkeysvalueof{/tikz/grid with coordinates/minor step}
	},
	major help lines/.style={
		help lines,
		line width=\pgfkeysvalueof{/tikz/grid with coordinates/major line width},
		step=\pgfkeysvalueof{/tikz/grid with coordinates/major step}
	},
	grid with coordinates/.cd,
	minor step/.initial=.1,
	major step/.initial=.5,
	major line width/.initial=0.7pt,
}

\title{Unit Disk Visibility Graphs\footnote{This work is supported by the {Czech Science Foundation}, project no.~{20-04567S}.}}
\author{Deniz A\u{g}ao\u{g}lu\thanks{Faculty of Informatics, Dept. of Computer Science, Masaryk University} \and Onur \c{C}a\u{g}{\i}r{\i}c{\i}\footnotemark[2]}
\begin{document}

	
	

	\maketitle
	
	\begin{abstract}
	We study unit disk visibility graphs, where the visibility relation between a pair of geometric entities is defined by not only obstacles, but also the distance between them.
	That is, two entities are not mutually visible if they are too far apart, regardless of having an obstacle between them.
	This particular graph class models real world scenarios more accurately compared to the conventional visibility graphs. We first define and classify the unit disk visibility graphs, and then show that the 3-coloring problem is NP-complete when unit disk visibility model is used for a set of line segments (which applies to a set of points) and for a polygon with holes.
	\end{abstract}

	\section{Introduction} \label{sec:intro}
	
	A visibility graph is a simple graph $G = (V,E)$ defined over a set $\mathcal{P} = \{p_1, \dots, p_n\}$ of $n$ geometric entities (points, vertices of a polygon, endpoints of a set of line segments)  where a vertex $u \in V$ represents a geometric entity $p_u \in \mathcal{P}$, and the edge $uv \in E$ exists if and only if $p_u$ and $p_v$ are mutually visible (or see each other).
	In the literature, visibility graphs were studied considering various geometric sets such as a simple polygon \cite{Rourke_artgallery}, a polygon with holes \cite{Wein_voronoi}, a set of points \cite{Cardinal_pointcomplexity}, a set of line segments \cite{Everett_planarsegment}, along with different visibility models such as line-of-sight visibility \cite{Garey_lineOfSight}, $\alpha$-visibility \cite{Ghodsi_alpha}, and $\pi$-visibility \cite{Urrutia_artGalleryAndIllum}.
	
	Visibility graphs are used to describe real-world scenarios majority of which concern the mobile robots and path planning \cite{Latombe_robotmotion,Berg_compgeo,Oommen_robotnavigation}.
	While modeling the environment in which the robots move, a very common tool is to interpret the geometric entities and the relations among them as visibility graphs \cite{ORourke_handbook,Aichholzer_convexifying,Overmars_newmethods,Floriani_visibilityonterrain}.
	
	However, the physical limitations of the real world are usually overlooked or ignored while using visibility graphs. Since no camera, sensor, or guard (the objects represented by vertices of the visibility graph) has infinite range, two objects might not sense each other even though there are no obstacles in-between. Based on such a limitation, we assume that if a pair of objects  are too far from each other, then they do not see each other.	To model this notion, we adapt the unit disk graph model.
	
	\begin{figure}[htbp]
		\centering
		\captionsetup[subfigure]{position=b,justification=centering}
		\begin{subfigure}[t]{0.24\linewidth}
			\begin{tikzpicture}[scale=2]
			\draw (0,0) to[grid with coordinates] (1.5,1);
			
			\tikzstyle{every node}=[draw=black, fill=black, shape=circle, minimum size=3pt,inner sep=0pt];
			
			\node (1) at (0,0) {};
			\node (2) at (0.5,0.6) {};
			\node (3) at (0.8,1) {};
			\node (4) at (1.5,0.5) {};		
			\node (5) at (1.4,0.1) {};
			\node (6) at (0.9,0.2) {};
			
			\tikzstyle{every path}=[red]
			
			\draw (1)--(2)--(3)--(4)--(5);
			\draw (2)--(4);
			
			\foreach \i in {1,...,5} 
			{
				\draw (\i)--(6);
			}
			
			\tikzstyle{every path}=[black, dashed]	
			\draw (1)--(4);
			\draw (1)--(5)--(2);
			\draw (3)--(5);
			\draw (1)--(3);
			
			\end{tikzpicture}
			\caption{}
			\label{fig:UDVGa}
		\end{subfigure}
		~
		\begin{subfigure}[t]{0.24\linewidth}
			\begin{tikzpicture}[scale=2]
			\draw (0,0) to[grid with coordinates] (1.5,1);
			
			\tikzstyle{every node}=[draw=black, fill=black, shape=circle, minimum size=3pt,inner sep=0pt];
			
			\node (1) at (0,0) {};
			\node (2) at (0.5,0.6) {};
			\node (3) at (0.8,1) {};
			\node (4) at (1.5,0.5) {};		
			\node (5) at (1.4,0.1) {};
			\node (6) at (0.9,0.2) {};
			
			\tikzstyle{every path}=[thick]
			\draw (1)--(2);
			\draw (3)--(6);
			\draw (4)--(5);
			
			\tikzstyle{every path}=[red]
			
			\draw (2)--(3)--(4);

			\foreach \i in {1,2,4,5} 
			{
				\draw (\i)--(6);
			}
			
			\tikzstyle{every path}=[black, dashed]	
			
			\draw (1)--(5);
			\draw (3)--(5);
			\draw (1)--(3);
			
			\end{tikzpicture}
			\caption{}
			\label{fig:UDVGb}
		\end{subfigure}
		~
		\begin{subfigure}[t]{0.24\linewidth}
			\begin{tikzpicture}[scale=2]
			\draw (0,0) to[grid with coordinates] (1.5,1);
			
			\tikzstyle{every node}=[draw=black, fill=black, shape=circle, minimum size=3pt,inner sep=0pt];
			
			\node (1) at (0,0) {};
			\node (2) at (0.5,0.6) {};
			\node (3) at (0.8,1) {};
			\node (4) at (1.5,0.5) {};		
			\node (5) at (1.4,0.1) {};
			\node (6) at (0.9,0.2) {};
			
			\fill[gray, opacity=0.3] (1.center)--(2.center)--(3.center)--(4.center)--(5.center)--(6.center)--(1.center);
			
			\tikzstyle{every path}=[thick]
			
			\draw (1)--(2)--(3)--(4)--(5)--(6)--(1);
			
			\tikzstyle{every path}=[red]
			
			\draw (2)--(4);
			\draw (2)--(6);
			\draw (3)--(6);
			\draw (4)--(6);
			
			\tikzstyle{every path}=[black, dashed]	
			
			\draw (1)--(4);
			\draw (2)--(5);
			\draw (3)--(5);
			
			\end{tikzpicture}
			\caption{}
			\label{fig:UDVGc}
		\end{subfigure}
		~
		\begin{subfigure}[t]{0.24\linewidth}
			\begin{tikzpicture}[scale=2]
			\draw (0,0) to[grid with coordinates] (1.5,1);
			
			\tikzstyle{every node}=[draw=black, fill=black, shape=circle, minimum size=3pt,inner sep=0pt];
			
			\node (1) at (0,0) {};
			\node (2) at (0.5,0.6) {};
			\node (3) at (0.8,1) {};
			\node (4) at (1.5,0.5) {};		
			\node (5) at (1.4,0.1) {};
			\node (6) at (0.9,0.2) {};
			
			\node (h1) at (0.5,0.4) {};
			\node (h2) at (0.8,0.7) {};
			\node (h3) at (1,0.5) {};
			\node (h4) at (0.7,0.4) {};
			
			\fill[gray, opacity=0.3] (1.center)--(2.center)--(3.center)--(4.center)--(h3.center)--(h2.center)--(h1.center)--(1.center);
			\fill[gray, opacity=0.3] (4.center)--(5.center)--(6.center)--(1.center)--(h1.center)--(h4.center)--(h3.center)--(4.center);		
			
			\tikzstyle{every path}=[thick]
			\draw (h1)--(h2)--(h3)--(h4)--(h1);
			
			\draw (1)--(2)--(3)--(4)--(5)--(6)--(1);
			
			\tikzstyle{every path}=[red]
			
			\draw (1)--(h1);
			\draw (1)--(h4);
			
			\draw (2)--(h1);
			\draw (2)--(h2);
			
			\draw (3)--(h2);
			\draw (3)--(h3);	
			
			\draw (4)--(h2);		
			\draw (4)--(h3);
			\draw (4)--(h4);
			\draw (4)--(6);
			
			\draw (5)--(h3);
			\draw (5)--(h4);
			
			\draw (6)--(h1);
			\draw (6)--(h3);
			\draw (6)--(h4);	
			
			\tikzstyle{every path}=[black, dashed]	
			
			\draw (1)--(4);
			\draw (1)--(h2);
			\draw (3)--(5);

			\end{tikzpicture}
			\caption{}
			\label{fig:UDVGd}
		\end{subfigure}
		\caption{Unit disk visibility relations of (a) a set of points, (b) a set of line segments, (c) a simple polygon, and (d) a polygon with a hole.}
		
		\label{fig:UDVG}
	\end{figure}
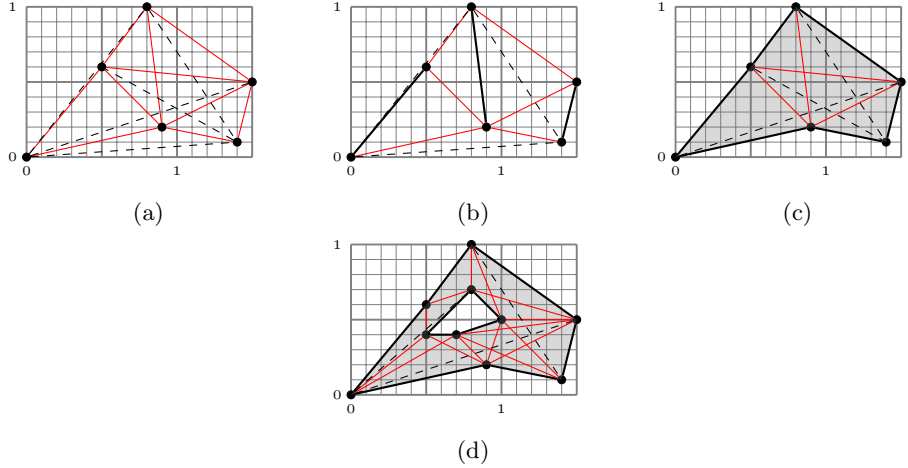
	
	$G$ is called a \emph{unit disk visibility graph} of $\mathcal{P}$ if the existence of an edge $uv \in E$ means that $p_u$ and $p_v$ see each other, and the Euclidean distance between them is less than 1 unit.
	In other words, an element of the geometric set cannot see another element if they are too far apart\footnote{In a unit disk intersection graph, there exists an edge between two vertices, if the corresponding disks intersect. This means that a distance $\leq 2$ is enough. However, our paper is based on wireless sensors and their ranges. For two sensor nodes to communicate, they must be inside each other's communication range. Thus, instead of the radius, we assume that the diameter of a unit disk is one unit.}.
	Unit disk point visibility graphs directly follows from this definition while for unit disk segment and polygon visibility graphs, the additional constraints are the followings: $i)$ the edges of unit disk segment visibility graphs cannot to intersect any segment, and $ii)$ the edges of unit disk polygon visibility graphs must totally lie inside the polygon.
	
	In this paper, we tackle the 3-coloring problem in unit disk visibility graphs.
	The 3-coloring problem is a famous combinatorial decision problem which asks if a graph has a \emph{proper 3-coloring} (or in short, \emph{3-coloring}).
	A graph is said to have a 3-coloring (or to be 3-colorable) when all the vertices receive one of the three pre-given colors, and no two adjacent vertices receive the same color.
	This problem is NP-complete for graphs in general \cite{chroComp}, and in this paper we show that it is also NP-complete on unit disk visibility graphs of a set of line segments, and a polygon with holes. We refer the reader to the Appendix for the full proofs of axioms marked with ~\apxmark. 
	
	\section{Preliminary results} \label{sec:classification}
	
	In this section, we first show that visibility graphs are a proper subclass of unit disk visibility graphs. We assume that the given geometric set is a set of points, since every visibility graph considered in this paper can be embedded in the Euclidean plane; points, endpoints of a set of line segments, and the vertices of a polygon.

		\subsection{The classification of the visibility graphs and the unit disk visibility graphs}

	\begin{lemma}\apxmark \label{lem:scaledown}
		Consider a set $P = \{p_1, \dots, p_n\}$ of points, and the visibility graph $G(P)$ of $P$.
		There exists an embedding $\Sigma(P)$ of $P$, such that the Euclidean distance between every pair $p,q \in P$ is less than one unit, preserving the visibility relations.
	\end{lemma}
	\begin{proof}
		Let $\Phi = \{x_1, y_1, x_2, y_2, \dots, x_n, y_n\}$ be the set of all coordinates used to represent the set $P$, and let $\varphi \in \Phi$ be a coordinate whose absolute value is the largest number in $\Phi$.
		
		The circle $C$ with center $(0,0)$ and radius $\varphi$ contains all the points in $P$. 
		Now imagine we shrink $C$ into a unit circle, and scale the point set accordingly.
		In order to do that, we divide every coordinate in $\Phi$ by $2\varphi$. 
		That is, the new coordinates for the points are $(x_1/2\varphi, y_1/2\varphi), \dots, (x_n/2\varphi, y_n/2\varphi)$.
		
		Considering any pair $s_i, s_j \in S$, and the straight line $\ell(i,j)$ passing through these points, the slope of $\ell$ does not change. 
		Therefore, the visibility relations are preserved for $P$.
	\end{proof}
	
	By Lemma~\ref{lem:scaledown}, a given set $P$ of points can be scaled down to obtain $P'$ so that every point in $P'$ is inside a unit circle, and the visibility graph $G(P)$ of $P$ is exactly the same as the visibility graph $G(P')$ of $P'$.
	Thus, the unit disk visibility graph of $P'$ is also isomorphic to $G(P')$ since no pair of points in $P'$ has Euclidean distance greater than one unit. We easily get the following.
	
	\begin{lemma}\apxmark \label{lem:hardUDVG}
		If a problem $\mathfrak{Q}$ is NP-hard for point visibility graphs, then $\mathfrak{Q}$ is also NP-hard for unit disk point visibility graphs.
	\end{lemma}
	\begin{proof}
		Let $P = \{p_1, \dots, p_n\}$ be a geometric set (points, vertices of a polygon, endpoints of a set of line segments) in the Euclidean plane, with coordinates $(x_1, y_1)$, $\dots$, $(x_n, y_n)$, respectively.
		Let $\mathcal{A}$ be an algorithm that solves an instance $\mathcal{Q}$ of the problem $\mathfrak{Q}$ on $P$, in polynomial time.
		
		As shown in Lemma~\ref{lem:scaledown}, $P$ can be scaled down into a unit circle, preserving the relations.
		To avoid the possible coordinates with exponentially many digits, instead of making the transformation $(x_i, y_i) \to (x_i/2\varphi, y_i/2\varphi)$ for each $(x_i, y_i)$, we let $M$ be the smallest integer larger than $|\varphi|$, which can be expressed as $2^k$.
		Since dividing a number $n$ by $2^k$ requires at most $k = O(\log n)$ many digits, applying the transformation $(x_i, y_i) \to (x_i/2M, y_i/2M)$ does not create coordinates with exponentially many digits.
		
		So, let $P' = \{p'_1, \dots, p'_n\}$ be a point set with coordinates $(x_1/2M, y_1/2M), \dots, (x_n/2M,$ $y_n/2M)$, respectively.
		If $\mathcal{A}$ solves the problem $\mathfrak{Q}$ on $P'$ in polynomial time, then $\mathcal{A}$ solves $\mathfrak{Q}$ on $P$ also in polynomial time.
		
		Therefore, any $\mathfrak{Q}$ is also NP-hard for unit disk point visibility graphs.
	\end{proof}
	
	\begin{rem}
		By Lemma~\ref{lem:hardUDVG}, the maximum dominating set, the minimum vertex cover, the maximum independent set and the minimum dominating set problems which have been shown to be NP-hard for visibility graphs by \cite{Lin_complexityaspects, LeeLin_artGallery} are also NP-hard for unit disk visibility graphs.
	\end{rem}

	\begin{figure}[htbp]
		\captionsetup[subfigure]{position=b,justification=centering}
		\centering
		\begin{subfigure}[b]{0.23\linewidth}
			\centering
			\begin{tikzpicture}[scale=0.8]
			\tikzstyle{every node}=[draw=black, fill=gray, shape=circle, minimum size=3pt,inner sep=0pt];
			\node[fill=red] (center) at (0,0) {};

			\foreach \i in {1,...,11} {
				\coordinate (\i) at (\i*360/11:1cm);
			}
			
			\draw (1) -- (11);
			\draw (center)--(11);
			\foreach \i in {3,...,11} {
				\pgfmathtruncatemacro\j{\i-1};
				\draw (\i) -- (\j);
				\draw (center)--(\j);
			}
			
			\foreach \i in {1,2,4,6,8,10} \draw[thick, red] (center)--(\i);
			
			\foreach \i in {3,5,7,9,11} \node at (\i) {};
			
			\tikzstyle{every node}=[draw=red, fill=red, shape=circle, minimum size=3pt,inner sep=0pt];
			\foreach \i in {1,2,4,6,8,10} \node at (\i) {};
			
			\end{tikzpicture}
			\caption{}
			\label{fig:2a}
		\end{subfigure}
		~
		\begin{subfigure}[b]{0.23\linewidth}
			\centering
			\begin{tikzpicture}[scale=3]
			\node[draw=none, fill=none, rotate=26] at (0.78,0) {\tiny $1$ unit};
			
			\node[draw=none, fill=none] at (0.5,-0.07) {\small $u$};
			
			\tikzstyle{every node}=[draw=black, fill=gray, shape=circle, minimum size=3pt,inner sep=0pt];
			\draw[<->] (0.55,-0.05) -- (0.985,0.145);
			
			\node (1) at (0,0.2) {};
			\node (2) at (0.1,0.2) {};
			\node (3) at (0.2,0.2) {};
			\node (4) at (0.3,0.2) {};
			\node (5) at (0.4,0.2) {};
			\node (6) at (0.5,0.2) {};
			\node (7) at (0.6,0.2) {};
			\node (8) at (0.7,0.2) {};
			\node (9) at (0.8,0.2) {};
			\node (10) at (0.9,0.2) {};
			
			\node (11) at (1,0.2) {};
			
			\node (12) at (0.5,0) {};
			
			\foreach \i in {1,...,10}
			{
				\pgfmathtruncatemacro\j{\i+1};
				\draw (\i) -- (\j);
				\draw (\i) -- (12);
			}
			
			\draw (11) -- (12);
			
			\foreach \i in {1,3,5,7,9,11} \draw[thick, red] (12)--(\i);
			
			\end{tikzpicture}
			\caption{}
			\label{fig:2b}
		\end{subfigure}
		~
		\begin{subfigure}[b]{0.23\linewidth}
			\centering
			\begin{tikzpicture}[scale=3]
			\node[draw=none, fill=none, rotate=30] at (0.75,0) {\tiny $1$ unit};
			
			\node[draw=none, fill=none] at (0.4,-0.2) {\small $u$};
			\node[draw=none, fill=none] at 
			(0.4,-0.03) {\small $v$};
			
			\tikzstyle{every node}=[draw=black, fill=gray, shape=circle, minimum size=3pt,inner sep=0pt];
			\draw[<->] (0.55,-0.05) -- (0.95,0.15);
			
			\node (1) at (0,0.2) {};
			\node (2) at (0.1,0.2) {};
			\node (3) at (0.2,0.2) {};
			\node (4) at (0.3,0.2) {};
			\node (5) at (0.4,0.2) {};
			\node (6) at (0.5,0.2) {};
			\node (7) at (0.6,0.2) {};
			\node (8) at (0.7,0.2) {};
			\node (9) at (0.8,0.2) {};
			\node (10) at (0.9,0.2) {};
			\node (11) at (0.5,0) {};
			\node (12) at (0.5,-0.2) {};
			
			\foreach \i in {1,3,5,7,9,11}
			{
				\pgfmathtruncatemacro\j{\i+1};
				\draw[ultra thick] (\i) -- (\j);
				\draw (\i) -- (11);
				\draw (\j) -- (11);
			}
			
			\foreach \i in {2,4,6,8}
			{
				\pgfmathtruncatemacro\j{\i+1};
				\draw (\i) -- (\j);
			}
			
			\draw (10) -- (11);
			
			\foreach \i in {1,3,5,7,9,12} \draw[thick, red] (11)--(\i);
			
			\end{tikzpicture}
			\caption{}
			\label{fig:2c}
		\end{subfigure}
		~
		\begin{subfigure}[b]{0.23\linewidth}
			\centering
			\begin{tikzpicture}[scale=3]
			\node[draw=none, fill=none, rotate=45] at (0.8,0.15) {\tiny $1$ unit};
			
			\node[draw=none, fill=none] at (0.5,-0.07) {\small $u$};
			
			\tikzstyle{every node}=[draw=black, fill=gray, shape=circle, minimum size=3pt,inner sep=0pt];
			\draw[<->] (0.55,-0.05) -- (1.04,0.475);
			
			\node (1) at (0,0.5) {};
			\node (2) at (0.1,0.46) {};
			\node (3) at (0.2,0.425) {};
			\node (4) at (0.3,0.395) {};
			\node (5) at (0.4,0.375) {};
			\node (6) at (0.5,0.375) {};
			\node (7) at (0.6,0.395) {};
			\node (8) at (0.7,0.425) {};
			\node (9) at (0.8,0.46) {};
			\node (10) at (0.9,0.5) {};
			
			\node (11) at (1,0.54) {};
			
			\node (12) at (0.5,0) {};
			
			\foreach \i in {1,...,11}
			{
				\pgfmathtruncatemacro\j{\i+1};
				\draw[ultra thick] (\i) -- (\j);			
			}
			\draw[ultra thick] (1) -- (12);			
			\foreach \i in {2,...,10}
			{
				\draw (\i)--(12);
			}
			
			\foreach \i in {1,3,5,7,9,11} \draw[thick, red] (12)--(\i);
			
			\end{tikzpicture}
			\caption{}
			\label{fig:2d}
		\end{subfigure}
		
		\caption{\textsc{(a)} A graph with an induced $K_{1,6}$ (indicated with red edges). Unit disk visibility graphs for \textsc{(b)} a set of points, \textsc{(c)} a set of segments, and \textsc{(d)} a simple polygon, each containing an induced $K_{1,6}$.}
		\label{fig:notUD}
	\end{figure}
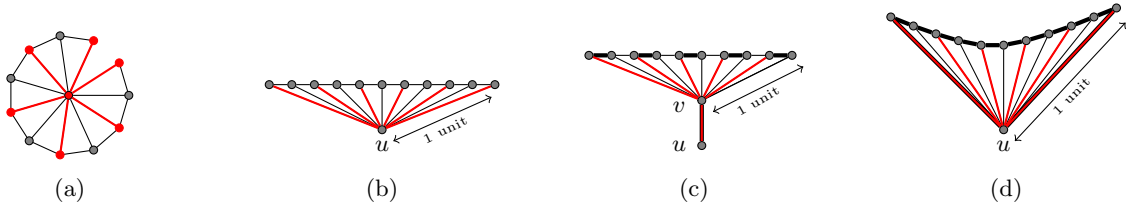			

	We now prove the classification of the unit disk graphs and the unit disk visibility graphs.

		\subsection{The classification of the unit disk graphs and the unit disk visibility graphs}

	\begin{lemma}\apxmark \label{lem:K1,6}
		Unit disk visibility graphs are not a subclass of unit disk graphs. 
	\end{lemma}
	\begin{proof}
		Unit disk graphs cannot contain an induced $K_{1,6}$ \cite{Marathe_heuristics} while a unit disk visibility graph of a set of points, a set of line segments, and a simple polygon can contain it.
		In Figure~\ref{fig:2a}, we have a graph which contains $K_{1,6}$ as an induced subgraph shown with red vertices and red edges. Our purpose is to show that an induced $K_{1,6}$ can be contained in the unit disk visibility graph of a set of points, a set of line segments and a simple polygon. In Figure~\ref{fig:2b}, we are given a set of twelve points where all of them except $u$ are collinear. Observe that among the collinear points, we can choose at most six pairwise non-adjacent points to be in our induced subgraph, otherwise at least two of them see each other. With the addition of $u$, we obtain an induced $K_{1,6}$ since the distance between $u$ and the farthest point is one unit. In Figure~\ref{fig:2c}, we have a similar configuration where we can choose at most five pairwise non-adjacent endpoints among the collinear segments. By adding both endpoints $u$ and $v$, we obtain an induced $K_{1,6}$ since $u$ does not see other points than $v$ due to distance and $v$ sees all endpoints. In Figure~\ref{fig:2d}, we also have a similar configuration to Figure~\ref{fig:2b}. The vertices of a polygon do not see each other if the straight line joining them is outside of the polygon. Therefore, among the vertices except $u$, we can choose at most six pairwise non-adjacent points to be in our induced subgraph. With the addition of $u$, we obtain an induced $K_{1,6}$.
	\end{proof}
	
	The idea used to prove Lemma~\ref{lem:K1,6} is that unit disk graphs cannot contain an induced $K_{1,6}$ \cite{Marathe_heuristics} while unit disk point, segment and polygon visibility graphs can contain it as in Figure~\ref{fig:notUD}. 
	
	
	\begin{lemma}\apxmark \label{lem:UnitUDVG}
		Unit disk graphs are a proper subclass of unit disk point visibility graphs, and \emph{not} a subclass of unit disk segment and polygon (simple or with holes) visibility graphs.
	\end{lemma}
	\begin{proof}
		Given a representation of unit disk graphs, we can simply perturb the disk centers slightly to obtain a set of points in general position, creating a configuration in which no three points are collinear \cite{Fonesca_recognition}.
		This way, we have had obtained a setting that is exactly a unit disk graph, in which a pair of vertices are adjacent only if they are close enough. 
		It is still a question whether the new positions of the disk centers can be represented by using polynomially many decimal digits with respect to the input size. However, this shows that every unit disk graph can be represented as the unit disk visibility graph of the new set of points. Lemma~\ref{lem:UnitUDVG} shows that there are unit disk visibility graphs of a set of points which cannot be recognized as unit disk graphs. Therefore, unit disk graphs are a ``proper'' subclass of unit disk visibility graphs for a set of points.
		
		By definition, a set of $n$ (disjoint) line segments contains $2n$ endpoints. Therefore, unit disk visibility graphs for a set of segments have even number of vertices while there unit disk graphs can have odd number of vertices. Combined with Lemma~\ref{lem:UnitUDVG}, unit disk graphs are neither a subclass nor a superclass of unit disk visibility graphs of a set of line segments.
		
		Unit disk visibility graphs for a simple polygon contain a Hamiltonian cycle which acts as the border of the given polygon. However, unit disk graphs may not contain an Hamiltonian cycle. Combined with Lemma~\ref{lem:UnitUDVG}, unit disk graphs are neither a subclass nor a superclass of unit disk visibility graphs for a simple polygon.
		
		Unit disk visibility graphs for a polygon with holes contain one cycle $C$ to act as the exterior border of the given polygon and at least one other cycle $C'$ disjoint from $C$ to act as the border of a hole. By definition, a cycle of a simple graph is of length at least three and this means that unit disk visibility graphs for a polygon with holes have at least $6$ vertices. Therefore, unit disk graphs of order $< 6$ are not unit disk visibility graphs for a polygon with holes. We now show that unit disk graphs of order $\geq 6$ are not unit disk visibility graphs for a polygon with holes using the following fact: \emph{The cycle $C'$ corresponding to the border of a hole is an induced ``chordless'' cycle of the graph.} However, unit disk graphs may not contain any induced cycle of length $\geq 3$. Combined with Lemma~\ref{lem:UnitUDVG}, unit disk graphs are neither a subclass nor a superclass of unit disk visibility graphs for a polygon with holes.	
	\end{proof}

	\section{Main results} \label{sec:3coloringSegment}
	
	In this section, we mention our NP-hardness reductions. A polynomial-time (NP-hardness) reduction from a (NP-hard) problem $\mathfrak{Q}$ to another problem $\mathfrak{P}$ is to map any instance $\Phi$ of $\mathfrak{Q}$ to some instance $\Psi$ of $\mathfrak{P}$ such that $\Phi$ is a \textsc{YES}-instance of $\mathfrak{Q}$ if and only if $\Psi$ is a \textsc{YES}-instance of $\mathfrak{P}$, in polynomial-time and polynomial-space.
	We first show that the 3-coloring problem for unit disk segment visibility graphs is NP-hard, using a reduction from the \emph{Monotone not-all-equal 3-satisfiability} (Monotone NAE3SAT) problem which is a variation of 3SAT \cite{Schaefer_complexitySAT} with no negated variables, and to satisfy the circuit, at least one variable must be true, and at least one variable must be false in each clause. 
	
	\def\polyV#1#2{%
		\begin{scope}[shift={#1}, rotate=#2]
			
			\coordinate (1) at (0,0);
			\coordinate (2) at (0.6, -0.8);
			\coordinate (3) at (0.8, -0.6);
			\coordinate (4) at (1.4, 0);
			\coordinate (5) at (0.8, 0.6);
			\coordinate (6) at (0.6, 0.8);
			\coordinate (7) at (0, 1.4);
			\coordinate (8) at (-0.6, 0.8);
			\coordinate (9) at (-0.8, 0.6);
			\coordinate (10) at (-1.4, 0);
			\coordinate (11) at (-0.8, -0.6);
			\coordinate (12) at (-0.6, -0.8);
			
			\draw[thick] (1)--(2);
			\draw[thick] (3)--(4)--(5);
			\draw[thick] (6)--(7)--(8);
			\draw[thick] (9)--(10)--(11);
			\draw[thick] (12)--(1);

			\fill[gray,opacity=0.3] (1) \foreach \i in {2,...,12}
			{
				--(\i)
			};		
			
		\end{scope}
	}
	
	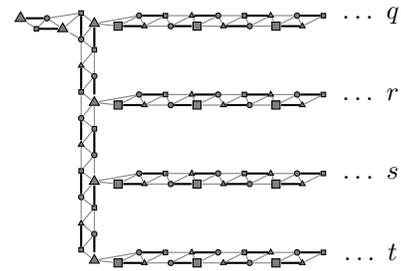
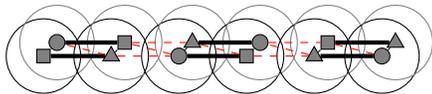
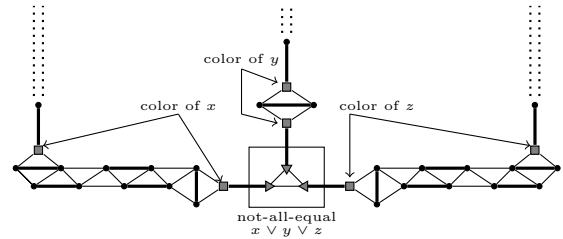
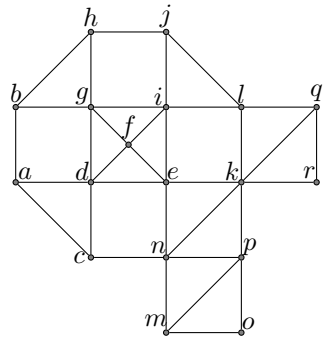
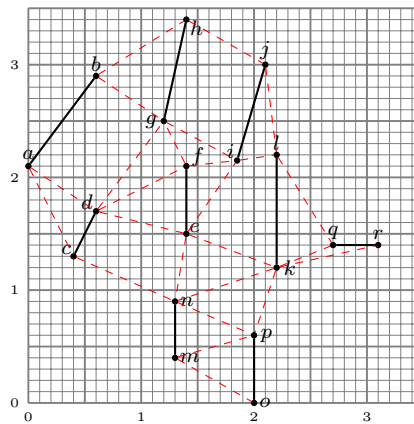
\begin{figure} [htbp]
		\captionsetup[subfigure]{position=b}
		\centering
		\begin{subfigure}[t]{0.48\linewidth}
			\centering
			\begin{tikzpicture}[yscale=0.43,
			triangle/.style = {regular polygon, regular polygon sides=3, scale=1.2},
			square/.style = {regular polygon, regular polygon sides=4, scale=1.2}]
			\tikzstyle{every node}=[draw, fill=black, minimum size=4pt,inner sep=0pt];
			
			\node[color=black, draw=none, fill=none] at (0,4) {$q$};
			\node[color=black, draw=none, fill=none] at (0,3) {$r$};
			\node[color=black, draw=none, fill=none] at (0,2) {$s$};
			\node[color=black, draw=none, fill=none] at (0,1) {$t$};
			
			\foreach \i in {1,2,3,4}
			{
				\draw (0.2,\i)--(6,\i);	
				\foreach \j in {0.5, 1.1, 1.7, 2.3, 2.9, 3.5, 4.1, 4.7, 5.3}
				{
					\node[color=gray,fill=gray!50, draw=gray, square] at (\j,\i) {};
					\node[color=gray,fill=gray!50, triangle] at (\j + 0.2, \i) {};
					\node[color=gray,fill=gray!50, shape=circle] at (\j + 0.4, \i) {};
				}	
			}
			
			\clause{0.5}{4}{2}{1}
			\clause{2.3}{4}{3}{1}
			\clause{4.1}{3}{2}{1}
			
			\tikzstyle{every node}=[draw, shape=rectangle, color=black, fill=none, minimum size=5pt, minimum width=4em, inner sep=2pt];
			
			\node at (1.1,-0.4) {$A$};
			\node at (2.9,-0.4) {$B$};
			\node at (4.7,-0.4) {$C$};
			\end{tikzpicture}
			\caption{ An NAE3SAT formula with variables $q,r,s,t$, and clauses 
				$A = (q \vee s \vee t)$,
				$B = (q \vee r \vee t)$, and 
				$C =(r \vee s \vee t)$.}
			\label{fig:3satCircuit}
		\end{subfigure}
		~
		\begin{subfigure}[t]{0.48\linewidth}
			\centering
			\begin{tikzpicture} [scale=0.35,
			triangle/.style = {regular polygon, regular polygon sides=3, scale=1.2},
			square/.style = {regular polygon, regular polygon sides=4, scale=1.2}]
			
			\node[draw=none, fill=none] at (11,-0.9) {$\dots\ q$};
			\node[draw=none, fill=none] at (11,-3.9) {$\dots\ r$};
			\node[draw=none, fill=none] at (11,-6.9) {$\dots\ s$};
			\node[draw=none, fill=none] at (11,-9.9) {$\dots\ t$};
			
			\tikzstyle{every node}=[draw=black, fill=gray, shape=circle, minimum size=2pt,inner sep=0pt];
			\tikzstyle{every path}=[draw, gray];

			\node[triangle,scale=2] (A) at (-2.3,-1) {};
			\node (B) at (-1.3,-1) {};
			\node[square] (C) at (-1.7,-1.4) {};
			\node[triangle,scale=1.8] (D) at (-0.7, -1.4) {};
			\draw[black, thick] (A)--(B);
			\draw (A)--(C);
			\draw (B)--(C);
			\draw[black, thick] (C)--(D);

			\foreach [evaluate={\j=int(mod(\i,3));}] \i in {1,...,10}
			{
				\ifthenelse{\j = 1}
				{
					\node[square] (\i) at (0, -\i+0.2) {};
					\node[triangle,scale=1.8] (\i') at (0.5, -\i-0.2) {};
				}
				{}
				
				\ifthenelse{\j = 2}
				{
					\node (\i) at (0, -\i+0.2) {};
					\node[square] (\i') at (0.5, -\i-0.2) {};
				}
				{}
				
				\ifthenelse{\j = 0}
				{
					\node[triangle] (\i) at (0, -\i+0.2) {};
					\node (\i') at (0.5, -\i-0.2) {};
				}
				{}

			}
			
			\draw (1)--(B)--(D)--(1);
			\draw (2)--(D);

			\foreach \i in {1,3,5,7,9}
			{
				\pgfmathtruncatemacro\j{\i+1};
				\draw[black, thick] (\i)--(\j);
				\draw (\i)--(\i');
				\draw[black, thick] (\i')--(\j');
				\draw (\i')--(\j);
				\draw (\j)--(\j');	
			}
			
			\foreach \i in {2,4,6,8}
			{
				\pgfmathtruncatemacro\j{\i+1};
				\draw (\i)--(\j);
				\draw (\i')--(\j');
				\draw (\i')--(\j);
			}

			\foreach [evaluate={\j=int(mod(\i,3));}] \i in {1,...,8}
			{
				\ifthenelse{\j = 1}
				{
					\node (qu\i) at (\i+1.2,-0.9) {};
					\node (ru\i) at (\i+1.2,-3.9) {};
					\node (su\i) at (\i+1.2,-6.9) {};
					\node (tu\i) at (\i+1.2,-9.9) {};
					
					\node[square,scale=1.8] (qd\i) at (\i+0.4,-1.3) {};
					\node[square,scale=1.8] (rd\i) at (\i+0.4,-4.3) {};
					\node[square,scale=1.8] (sd\i) at (\i+0.4,-7.3) {};	
					\node[square,scale=1.8] (td\i) at (\i+0.4,-10.3) {};
				}
				{}
				\ifthenelse{\j = 2}
				{
					\node[square] (qu\i) at (\i+1.2,-0.9) {};
					\node[square] (ru\i) at (\i+1.2,-3.9) {};
					\node[square] (su\i) at (\i+1.2,-6.9) {};
					\node[square] (tu\i) at (\i+1.2,-9.9) {};
					
					\node[triangle] (qd\i) at (\i+0.4,-1.3) {};
					\node[triangle] (rd\i) at (\i+0.4,-4.3) {};				
					\node[triangle] (sd\i) at (\i+0.4,-7.3) {};		
					\node[triangle] (td\i) at (\i+0.4,-10.3) {};
				}
				{}
				
				\ifthenelse{\j = 0}
				{
					\node[triangle] (qu\i) at (\i+1.2,-0.9) {};
					\node[triangle] (ru\i) at (\i+1.2,-3.9) {};
					\node[triangle] (su\i) at (\i+1.2,-6.9) {};
					\node[triangle] (tu\i) at (\i+1.2,-9.9) {};

					\node (qd\i) at (\i+0.4,-1.3) {};
					\node (rd\i) at (\i+0.4,-4.3) {};
					\node (sd\i) at (\i+0.4,-7.3) {};		
					\node (td\i) at (\i+0.4,-10.3) {};
				}
				{}

			}
			
			\foreach \i in {1,3,5,7}
			{
				\pgfmathtruncatemacro\j{\i+1};
				
				\draw[black, thick] (qu\i)--(qu\j);
				\draw[black, thick] (qd\i)--(qd\j);
				
				\draw[black, thick] (ru\i)--(ru\j);
				\draw[black, thick] (rd\i)--(rd\j);
				
				\draw[black, thick] (su\i)--(su\j);
				\draw[black, thick] (sd\i)--(sd\j);
				
				\draw[black, thick] (tu\i)--(tu\j);
				\draw[black, thick] (td\i)--(td\j);
				
				\draw (qu\i)--(qd\i);
				\draw (ru\i)--(rd\i);
				\draw (su\i)--(sd\i);
				\draw (tu\i)--(td\i);
				
				\draw (qu\j)--(qd\j);
				\draw (ru\j)--(rd\j);
				\draw (su\j)--(sd\j);
				\draw (tu\j)--(td\j);
				
				\draw (qu\i)--(qd\j);
				\draw (ru\i)--(rd\j);
				\draw (su\i)--(sd\j);
				\draw (tu\i)--(td\j);
				
			}

			\foreach \i in {2,4,6}
			{
				\pgfmathtruncatemacro\j{\i+1};
				
				\draw (qu\i)--(qu\j);
				\draw (qd\i)--(qd\j);
				
				\draw (ru\i)--(ru\j);
				\draw (rd\i)--(rd\j);
				
				\draw (su\i)--(su\j);
				\draw (sd\i)--(sd\j);
				
				\draw (tu\i)--(tu\j);
				\draw (td\i)--(td\j);

				\draw (qu\i)--(qd\j);
				\draw (ru\i)--(rd\j);
				\draw (su\i)--(sd\j);
				
			}
			
			\draw (qu1)--(1')--(qd1);
			\draw (ru1)--(4')--(rd1);
			\draw (su1)--(7')--(sd1);
			\draw (tu1)--(10')--(td1);

			\end{tikzpicture}
			\caption{The wires transfering one of two colors.}
			\label{fig:initial}
		\end{subfigure}
		
		\begin{subfigure}[t]{0.5\linewidth}
			\centering
			\begin{tikzpicture} [scale = 0.9,
			triangle/.style = {regular polygon, regular polygon sides=3, scale=1.2},
			square/.style = {regular polygon, regular polygon sides=4, scale=1.2}]
			
			\tikzstyle{every node}=[draw=black, fill=gray, shape=circle, minimum size=6pt,inner sep=0pt];
			
			\foreach \i in {1,2,3,4,5,6}
			{	
				\draw (\i,0) circle (0.55cm);
				\draw[gray] (\i+0.2,0.2) circle (0.55cm);
			}
			
			\node[square] (1) at (1,0) {};
			\node[triangle] (2) at (2,0) {};
			\node (3) at (3,0) {};
			\node[square] (4) at (4,0) {};
			\node[triangle] (5) at (5,0) {};
			\node (6) at (6,0) {};
			
			\node (7) at (1.2,0.2) {};
			\node[square] (8) at (2.2,0.2) {};
			\node[triangle] (9) at (3.2,0.2) {};
			\node (10) at (4.2,0.2) {};
			\node[square] (11) at (5.2,0.2) {};
			\node[triangle] (12) at (6.2,0.2) {};
			
			\foreach \i in {1,3,5}
			{
				\pgfmathtruncatemacro\j{\i+1};
				\pgfmathtruncatemacro\k{\i+6};
				\pgfmathtruncatemacro\m{\i+7};
				\draw[ultra thick] (\i) -- (\j);
				\draw[ultra thick] (\k) -- (\m);
			}
			
			\foreach \i in {2,3,4,5,6}
			{
				\pgfmathtruncatemacro\j{\i+5};
				\pgfmathtruncatemacro\k{\i+6};
				\draw[red, dashed] (\i) -- (\j);
				\draw (\i) -- (\k);
				
			}
			
			\foreach \i in {2,4,8,10}
			{
				\pgfmathtruncatemacro\j{\i+1};
				\draw[red, dashed] (\i) -- (\j);	
				
			}
			
			\draw (1) -- (7);
			
			\end{tikzpicture}
			
			\caption{A long edge gadget constructed by six line segments having a unique 3-coloring.}	
			\label{fig:longEdge}	
		\end{subfigure}
		~
		\begin{subfigure}[t]{0.45\linewidth}
			\centering
			\begin{tikzpicture}[scale=1.2, 
			square/.style = {regular polygon, regular polygon sides=4, fill=gray, scale=2.2},
			triangle/.style = {regular polygon, regular polygon sides=3, fill=gray, scale=2.2}
			]
			
			\tikzstyle{every node}=[draw=black, fill=black, shape=circle, minimum size=2pt,inner sep=0pt];
			
			\foreach \i in {0,1,2,3}
			{	
				\node (\i1) at (\i*0.5,0.4) {}; 
				\node (\i0) at (\i*0.5 + 0.2,0.2) {};
			}

			\node (A) at (2,0.4) {};
			\node (B)  at (2,0) {};
			\node[square] (C) at (2.3,0.2) {};
			\node (D) at (2.8,0.2) {};
			\node (E) at (3,0.4) {};
			\node[square] (F) at (3,0.9) {};
			\node (G) at (3.2,0.2) {};
			\node[square] (H) at (3.7,0.2) {};
			\node (I) at (4,0.4) {};
			\node (J)  at (4,0) {};
			
			\foreach \i in {8,9,10,11}
			{	
				\node (\i1) at (\i*0.5+ 0.5,0.4) {}; 
				\node (\i0) at (\i*0.5 + 0.3,0.2) {};
			}
			
			\foreach \i in {0,2,8,10}
			{
				\pgfmathtruncatemacro\j{\i+1};
				\draw[very thick] (\i0)--(\j0);
				\draw[very thick] (\i1)--(\j1);
			}
			
			\draw[very thick] (A)--(B);		
			\draw[very thick] (C)--(D);
			\draw[very thick] (E)--(F);
			\draw[very thick] (G)--(H);
			\draw[very thick] (I)--(J);

			\node[square] (P) at (0.25, 0.6) {};
			\node (Q) at (0.25, 1.1) {};
			\node (U) at (2.7, 1.1) {};
			\node (V) at (3.3, 1.1) {};
			\node[square] (R) at (3,1.3) {};
			\node (S) at (3,1.8) {};
			\node[square] (X) at (5.75, 0.6) {};
			\node (Y) at (5.75, 1.1) {};
			
			\draw[very thick] (P)--(Q);
			\draw[very thick] (U)--(V);
			\draw[very thick] (R)--(S);
			\draw[very thick] (X)--(Y);	
			
			\node[fill=none, shape=rectangle, minimum width=1cm, minimum height=0.8cm] at (3,0.3) {};

			\foreach \i in {1,9}
			{	
				\pgfmathtruncatemacro\j{\i+1};
				\pgfmathtruncatemacro\k{\i-1};
				\draw (\i0)--(\i1);
				\draw (\j0)--(\j1);
				\draw (\k0)--(\k1);
				\draw (\i0)--(\j0);
				\draw (\i1)--(\j1);
			}
			
			\draw (00)--(01);
			\draw (00)--(11);
			\draw (00)--(01);
			\draw (10)--(21);
			\draw (20)--(31);
			\draw (30)--(31);
			
			\draw (81)--(90);
			\draw (91)--(100);
			\draw (101)--(110);
			\draw (111)--(110);
			
			\draw (01)--(P)--(11);
			\draw (B)--(30)--(A)--(31);
			\draw (A)--(C)--(B);
			\draw (U)--(R)--(V);
			\draw (U)--(F)--(V);
			\draw (J)--(H)--(I);
			
			\draw (J)--(80)--(I)--(81);
			\draw (101)--(X)--(111);

			\draw (D)--(E)--(G)--(D);
			\node[triangle, rotate=-90] at (D.center) {};
			\node[triangle, rotate=180] at (E.center) {};
			\node[triangle,rotate=90] at (G.center) {};

			\draw[thick, dotted] (0.2,1.2)--(0.2,2.2);
			\draw[thick, dotted] (0.3,1.2)--(0.3,2.2);
			\draw[thick, dotted] (2.95,1.9)--(2.95,2.2);
			\draw[thick, dotted] (3.05,1.9)--(3.05,2.2);
			\draw[thick, dotted] (5.7,1.2)--(5.7,2.2);
			\draw[thick, dotted] (5.8,1.2)--(5.8,2.2);
			
			\node[draw=none, fill=none] at (1.8,1.1) {\tiny color of $x$};
			\node[draw=none, fill=none] at (2.5,1.6) {\tiny color of $y$};
			\node[draw=none, fill=none] at (4,1.1) {\tiny color of $z$};
			\node[draw=none, fill=none] at (3,-0.15) {\tiny not-all-equal};
			\node[draw=none, fill=none] at (3,-0.3) {\tiny $x \vee y \vee z$};
			\draw [->] (1.8,1)--(2.275,0.25);
			\draw [->] (1.8,1)--(0.3,0.65);
			
			\draw [->] (3.7,1)--(3.7,0.3);
			\draw [->] (3.7,1)--(5.7,0.65);
			
			\draw [->] (2.5,1.5)--(2.9,1.35);
			\draw (2.5,1.5)--(2.5,1);
			\draw[->] (2.5,1)--(2.9,0.9);
			\end{tikzpicture}
			\caption{The clause gadget for $x \vee y \vee z$.}
			\label{fig:naeclause}
		\end{subfigure}	
		
		\begin{subfigure}[b]{0.5\linewidth}
			\centering
			\begin{tikzpicture}[scale=0.5]
			\tikzstyle{every node}=[draw, fill=gray, shape=circle, minimum size=2pt,inner sep=0pt]
			
			\node[label=35:$a$] (A) at (0,2) {}; 
			\node[label=$b$] (B) at (0,4) {}; 
			\node[label=left:$c$] (C) at (2,0) {}; 
			\node[label=100:$d$] (D) at (2,2) {}; 
			\node[label=80:$e$] (E) at (4,2) {}; 
			\node[label=$f$] (F) at (3,3) {}; 
			\node[label=100:$g$] (G) at (2,4) {}; 
			\node[label=$h$] (H) at (2,6) {}; 
			\node[label=100:$i$] (I) at (4,4) {}; 
			\node[label=$j$] (J) at (4,6) {}; 
			\node[label=100:$k$] (K) at (6,2) {}; 
			\node[label=$l$] (L) at (6,4) {}; 
			\node[label=100:$m$] (M) at (4,-2) {}; 
			\node[label=100:$n$] (N) at (4,0) {}; 
			\node[label=80:$o$] (O) at (6,-2) {}; 
			\node[label=80:$p$] (P) at (6,0) {}; 
			\node[label=$q$] (Q) at (8,4) {}; 
			\node[label=100:$r$] (R) at (8,2) {}; 
			
			\draw (A)--(B);
			\draw (A)--(C);
			\draw (A)--(D);
			
			\draw (B)--(G);
			\draw (B)--(H);
			
			\draw (C)--(D);
			\draw (C)--(N);
			
			\draw (D)--(E);
			\draw (D)--(F);
			\draw (D)--(G);
			
			\draw (E)--(F);
			\draw (E)--(K);
			\draw (E)--(N);
			\draw (E)--(I);
			
			\draw (F)--(G);
			\draw (F)--(I);
			
			\draw (G)--(H);
			\draw (G)--(I);
			
			\draw (H)--(J);
			
			\draw (I)--(J);
			\draw (I)--(L);
			
			\draw (J)--(L);
			
			\draw (K)--(L);
			\draw (K)--(N);
			\draw (K)--(P);
			\draw (K)--(Q);
			\draw (K)--(R);
			
			\draw (L)--(Q);
			
			\draw (M)--(P);
			\draw (M)--(O);
			\draw (M)--(N);
			
			\draw (N)--(P);
			
			\draw (O)--(P);
			
			\draw (Q)--(R);
			
			\end{tikzpicture}
			\caption{The edge crossing gadget transferring the color from $a$ to $r$, and from $h$ to $o$.}
			\label{fig:crossingSegments}
		\end{subfigure}
		~
		\begin{subfigure}[b]{0.55\linewidth}
			\centering
			\hspace{-0.5em}
			\begin{tikzpicture}[scale=1.5]
			\draw (0,0) to[grid with coordinates] (3.5,3.5);
			
			\tikzstyle{every node}=[draw, fill=black, shape=circle, minimum size=2pt,inner sep=0pt];
			
			\node[label={\footnotesize $a$}] (A) at (0,2.1) {}; 
			\node[label={\footnotesize $b$}] (B) at (0.6,2.9) {}; 
			
			\node[label=135:{\footnotesize $c$}] (C) at (0.4,1.3) {}; 
			\node[label=135:{\footnotesize $d$}] (D) at (0.6,1.7) {}; 
			
			\node[label=20:{\footnotesize $e$}] (E) at (1.4,1.5) {}; 
			\node[label=20:{\footnotesize $f$}] (F) at (1.4,2.1) {}; 
			
			\node[label=left:{\footnotesize $g$}] (G) at (1.2,2.5) {}; 
			\node[label=-10:{\footnotesize $h$}] (H) at (1.4,3.4) {}; 
			
			\node[label=100:{\footnotesize $i$}] (I) at (1.85,2.15) {}; 
			\node[label=above:{\footnotesize $j$}] (J) at (2.1,3) {}; 
			
			\node[label=0:{\footnotesize $k$}] (K) at (2.2,1.2) {}; 
			\node[label={\footnotesize $l$}] (L) at (2.2,2.2) {}; 
			
			\node[label=0:{\footnotesize $m$}] (M) at (1.3,0.4) {}; 
			\node[label=0:{\footnotesize $n$}] (N) at (1.3,0.9) {}; 
			
			\node[label=0:{\footnotesize $o$}] (O) at (2,0) {}; 
			\node[label=0:{\footnotesize $p$}] (P) at (2,0.6) {}; 
			
			\node[label={\footnotesize $q$}] (Q) at (2.7,1.4) {}; 
			\node[label={\footnotesize $r$}] (R) at (3.1,1.4) {}; 

			\draw[thick] (A)--(B);
			\draw[thick] (C)--(D);
			\draw[thick] (E)--(F);
			\draw[thick] (G)--(H);
			\draw[thick] (I)--(J);
			\draw[thick] (K)--(L);
			\draw[thick] (M)--(N);
			\draw[thick] (O)--(P);
			\draw[thick] (Q)--(R);

			\tikzstyle{every path}=[dashed, draw, color=red];
			
			\draw (A)--(C);
			\draw (A)--(D);
			
			\draw (B)--(G);
			\draw (B)--(H);
			
			\draw (C)--(N);
			
			\draw (D)--(E);
			\draw (D)--(F);
			\draw (D)--(G);
			
			\draw (E)--(K);
			\draw (E)--(N);
			\draw (E)--(I);
			
			\draw (F)--(G);
			\draw (F)--(I);
			
			\draw (G)--(I);
			
			\draw (H)--(J);
			
			\draw (I)--(L);
			
			\draw (J)--(L);
			
			\draw (K)--(N);
			\draw (K)--(P);
			\draw (K)--(Q);
			\draw (K)--(R);
			
			\draw (L)--(Q);
			
			\draw (M)--(P);
			\draw (M)--(O);
			
			\draw (N)--(P);
			
			\end{tikzpicture}
			\caption{The segment embedding of ({\sc e}) on a grid. Only endpoint that is not on the grid is $i$, which is at $(1.85,2.15)$.}
			\label{fig:crossingGadgetSeg}
		\end{subfigure}
		
		\caption{The gadgets used in the proof of Theorem~\ref{thm:main}.}
		\label{fig:threeComponentsSeg}
	\end{figure}

		\subsection{The 3-coloring problem for unit disk segment visibility graphs}
		
		Before proving Theorem~\ref{thm:main}, we describe the gadgets used to construct a unit disk segment visibility graph from a given Monotone NAE3SAT formula in more detail, and show that they correctly transform an instance of the Monotone NAE3SAT problem to an instance of the 3-coloring problem for unit disk segment graphs.
		
		\subsubsection{The long edges and NAE3SAT clauses}\label{sec:longedges}
		
		First, let us describe how we model the long edges to transfer the colors from one vertex to the others.
		In Figure~\ref{fig:longEdge}, there are six line segments on two rows, drawn with thick lines.
		Their endpoints, which correspond to the vertices of the unit disk visibility graph, are shown by three different shapes: square, triangle, and circle.
		The disks drawn around the endpoints are the unit disks.
		The dashed red lines between the endpoints are the visibility edges.
		
		By definition, if two unit disks do not intersect, then the corresponding endpoints are not mutually visible.
		Thus, even though there is no obstacle between some endpoints, they do not share a visibility edge since they are farther than one unit.
		Note that if two endpoints see each other, then they are of different shape.
		Therefore, assigning a unique color for each unique shape gives us a proper 3-coloring in the unit disk visibility graph of this particular configuration.
		
		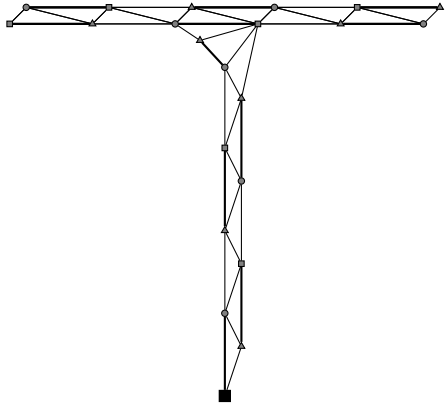
\begin{figure}[htbp]
			\centering
			\begin{tikzpicture} [scale=1.1,
			triangle/.style = {regular polygon, regular polygon sides=3, scale=0.5},
			square/.style = {regular polygon, regular polygon sides=4, scale=0.5},
			circle/.style = {shape=circle, scale=0.4}]
			
			\tikzstyle{every node}=[draw=black, fill=gray, shape=circle, minimum size=6pt,inner sep=0pt];

			\node[square] (1) at (1,0) {};
			\node[triangle] (2) at (2,0) {};
			\node[circle] (3) at (3,0) {};
			\node[square] (4) at (4,0) {};
			\node[triangle] (5) at (5,0) {};
			\node[circle] (6) at (6,0) {};
			
			\node[circle] (7) at (1.2,0.2) {};
			\node[square] (8) at (2.2,0.2) {};
			\node[triangle] (9) at (3.2,0.2) {};
			\node[circle] (10) at (4.2,0.2) {};
			\node[square] (11) at (5.2,0.2) {};
			\node[triangle] (12) at (6.2,0.2) {};
			
			\foreach \i in {1,3,5}
			{
				\pgfmathtruncatemacro\j{\i+1};
				\pgfmathtruncatemacro\k{\i+6};
				\pgfmathtruncatemacro\m{\i+7};
				\draw[thick] (\i) -- (\j);
				\draw[thick] (\k) -- (\m);
			}
			
			\foreach \i in {2,3,4,5,6}
			{
				\pgfmathtruncatemacro\j{\i+5};
				\pgfmathtruncatemacro\k{\i+6};
				\draw(\i) -- (\j);
				\draw (\i) -- (\k);
				
			}
			
			\foreach \i in {2,4,8,10}
			{
				\pgfmathtruncatemacro\j{\i+1};
				\draw (\i) -- (\j);	
				
			}
			
			\draw (1) -- (7);
			
			\node[triangle] (13) at (3.3, -0.2) {};
			\node[circle] (14) at (3.6, -0.525) {};
			
			\node[square] (15) at (3.6, -1.5) {};
			\node[triangle] (16) at (3.6, -2.5) {};
			
			\node[circle] (17) at (3.6, -3.5) {};
			\node[square, scale=2, fill=black] (18) at (3.6, -4.5) {};
			
			\node[triangle] (19) at (3.8,-0.9) {};
			\node[circle] (20) at (3.8,-1.9) {};
			
			\node[square] (21) at (3.8,-2.9) {};
			\node[triangle] (22) at (3.8,-3.9) {};
			
			\foreach \i in {1,3,5}
			{
				\pgfmathtruncatemacro\j{\i+1};
				\pgfmathtruncatemacro\k{\i+6};
				\pgfmathtruncatemacro\m{\i+7};
				\draw[thick] (\i) -- (\j);
				\draw[thick] (\k) -- (\m);
			}
			
			\foreach \i in {2,4,8,10}
			{
				\pgfmathtruncatemacro\j{\i+1};
				\draw (\i) -- (\j);
			}
			
			\foreach \i in {13,15,17,19,21}
			{
				\pgfmathtruncatemacro\j{\i+1};
				\draw[thick] (\i) -- (\j);	
			}
			
			\draw (1)--(7)--(2)--(8)--(3)--(9)--(4)--(10)--(5)--(11)--(6);
			\draw (3)--(13);
			\draw (4)--(19);
			\draw (13)--(4)--(14)--(15);
			\draw (14)--(19)--(15)--(20)--(16)--(21)--(17)--(22)--(18);
			\draw (16)--(17);
			\draw (20)--(21);
			\end{tikzpicture}
			\caption{Setting up the long edge gadget to transfer the color of square endpoints given in Figure~\ref{fig:3satCircuit} and Figure~\ref{fig:initial}.}
			\label{fig:transferringSegment}
		\end{figure}
		
		\begin{obs}
			In a proper 3-coloring of the unit disk visibility graph of the line segment configuration given in Figure~\ref{fig:longEdge}, a pair of vertices receive the same color if and only if their corresponding endpoints are of same shape.
		\end{obs}
		
		Observe that on each row, the shapes of the endpoints are repeating sequentially.
		Given such a configuration with arbitrarily many line segments, if we consider the ordering $(1,2, \dots)$ of the endpoints from left to right on a single row, then the color of $i$th endpoint will have the same color with $(i+3)$th endpoint.
		This helps us to transfer a color from one side of the configuration to the desired position, enabling the usage of long edges in the circuit.
		
		Consider the circuit given in Figure~\ref{fig:3satCircuit}.
		There are four boolean variables, $q$, $r$, $s$, $t$, and three clauses, $A$, $B$, $C$.
		Each variable has a long edge, representing its ``wire.''
		On each row, three shapes are repeating, square, triangle, and circle, in this precise order.
		These three shapes are in the same order with the configuration given in Figure~\ref{fig:longEdge}.
		If a variable appears in a clause, then it is shown by an edge from the corresponding row to the clause.
		Thus, when a color is assigned to a square endpoint, the same color repeats along the wire,
		and eventually is transferred to the clause.
		
		The truth assignments of the variables are determined by a pair of colors.
		In this case, every square endpoint on a wire should receive one of these two colors.
		We guarantee this by setting up the gadget shown in Figure~\ref{fig:initial}.
		The color of the triangle vertex is the ``neutral'' color, which means that the remaining two colors represents ``true'' and ``false'' for the variables.
		Thus, by picking a color for the triangle, we enforce every variable to transfer either true or false.
		Important part is that a variable connects to a clause only by one of the squares.
		Therefore, the truth assignment of a variable is transferred to the clause by the color of a square endpoint.
		Let us assume that the triangles in Figure~\ref{fig:initial} are colored green.
		In Figure~\ref{fig:transferringSegment}, we see an embedding of line segments, in which the variables transfer either blue, or red to the clauses below.
		
		Three components given in Figure~\ref{fig:3satCircuit}, Figure~\ref{fig:initial}, and Figure~\ref{fig:transferringSegment}
		guarantee that two colors representing true and false are transferred properly to the NAE3SAT clauses.
		In a NAE3SAT instance, a clause should be not-all-equal i.e., in a clause at least one variable must be false, and at least one variable must be true.
		Now, let us show that if a clause receives three variables with the same truth assignment, then the gadget is not 3-colorable.
		
		Regardless of the number of variables, and number of clauses, the long vertical edges in Figure~\ref{fig:3satCircuit} never intersect.
		This is because a variable transfers its color from different endpoints on the corresponding row, for every clause it appears in.
		In Figure~\ref{fig:3satCircuit}, we see that the variables are connected to clauses via vertical edges these are close to each other.
		In other words, every clause has its own area, and no two vertical edges intersect.
		This allows us to design a gadget that is able to connect three variables, even if they are far apart.
		
		Up to this point, we showed how we model the long edges to transfer the colors.
		However, given an instance of the Monotone NAE3SAT problem, there is no guarantee that the circuit can be drawn without edge crossings.
		Note that the circuit given in Figure~\ref{fig:3satCircuit}, the vertical edges passes through several long edges.
		Since our embedding is in 2D, we need a gadget which has no edge crossings, and also can be used to transfer the colors properly when two edges intersect in the circuit.
		
		\subsubsection{The edge crossings}\label{sec:edgecrossings}
		
		We now describe a certificate that transfers the colors safely in case of edge crossings.
		The certificate that we use is similar to the certificate given by Gr{\"a}f et al. for the reduction of unit disk graph coloring problem (see Figure 4 in \cite{Graf_udgColoring}).
		
		
		In Figure~\ref{fig:crossingSegments} and Figure~\ref{fig:crossingGadgetSeg}, we give the gadget to replace the edge crossings in the 3SAT circuit, and its embedding as line segments, respectively. As we have mentioned previously, the colors repeat.
		
		\begin{clm} \label{2uniqueCol}
			There are exactly two distinct (proper) 3-colorings of the edge crossing gadget gadget given in Figure~\ref{fig:crossingSegments}, all of which require the vertex pairs $a, r$ and $h, o$ to receive the same color.
		\end{clm}
		
		\begin{proof}
			Observe that there is a unique coloring of the induced subgraph $G'$ by the vertex subset $\{d, e, f, g, i\}$ up to a permutation of colors. Since $f$ is adjacent to all vertices in $G'$, it has to obtain a unique color $c_g$. Then, coloring any of $\{d, e, g, i\}$ uniquely determines the colors of the rest of these vertices where $e$ and $g$ get the same color $c_b \neq c_g$, and $d$ and $i$ get the same color $c_r \neq c_g, c_b$. We continue by choosing any of the triangles $\{e, k, n\}$, $\{i, j, l\}$, $\{a, c, d\}$  and $\{b, g, h\}$ whose one vertex is already colored by the coloring of $G'$. Assume that we choose $\{e, k, n\}$ where $e$ has the color $c_b$. There are two cases which may occur. 
			
			The first case is when $k$ is colored with $c_g$ and $n$ is colored with $c_r$.
			In this case, the color of every uncolored vertex is determined and as in Figure~\ref{fig:case1}. Note that this case is the same as the following cases: \emph{i)} choosing the triangle $\{i, j, l\}$ and coloring $j$ with $c_g$ and $l$ with $c_b$, \emph{ii)} choosing the triangle $\{a, c, d\}$ and coloring $a$ with $c_b$ and $c$ with $c_g$, and \emph{iii)} choosing the triangle $\{b, g, h\}$ and coloring $b$ with $c_g$ and $h$ with $c_r$.
			
			The second case is when $k$ is colored with $c_r$ and $n$ is colored with $c_g$.
			In this case, the color of every uncolored vertex is determined and as in Figure~\ref{fig:case2} (b). Note that this case is the same as the following cases: \emph{i)} choosing the triangle $\{i, j, l\}$ and coloring $j$ with $c_b$ and $l$ with $c_g$, \emph{ii)} choosing the triangle $\{a, c, d\}$ and coloring $a$ with $c_g$ and $c$ with $c_b$, and \emph{iii)} choosing the triangle $\{b, g, h\}$ and coloring $b$ with $c_r$ and $h$ with $c_g$.
			
			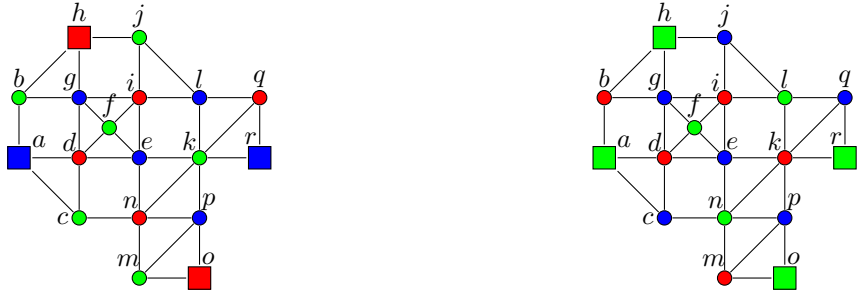
\begin{figure}[htbp]
				\centering
				\captionsetup[subfigure]{position=b}
				\begin{subfigure}[b]{0.45\linewidth}
					\centering
					\begin{tikzpicture}[scale=0.4,
					square/.style = {regular polygon, regular polygon sides=4, scale=4},
					circle/.style = {shape=circle, scale=1.8}]
					\tikzstyle{every node}=[draw, color=black, shape=circle, minimum size=3pt,inner sep=0pt]
					
					\node[label=35:$a$, square, fill=blue] (A) at (0,2) {}; 
					\node[label=$b$, circle, fill=green] (B) at (0,4) {}; 
					\node[label=left:$c$, circle, fill=green] (C) at (2,0) {}; 
					\node[label=100:$d$, circle, fill=red] (D) at (2,2) {}; 
					\node[label=80:$e$, circle, fill=blue] (E) at (4,2) {}; 
					\node[label=$f$, circle, fill=green] (F) at (3,3) {}; 
					\node[label=100:$g$, circle, fill=blue] (G) at (2,4) {}; 
					\node[label=$h$, square, fill=red] (H) at (2,6) {}; 
					\node[label=100:$i$, circle, fill=red] (I) at (4,4) {}; 
					\node[label=$j$, circle, fill=green] (J) at (4,6) {}; 
					\node[label=100:$k$, circle, fill=green] (K) at (6,2) {}; 
					\node[label=$l$, circle, fill=blue] (L) at (6,4) {}; 
					\node[label=100:$m$, circle, fill=green] (M) at (4,-2) {}; 
					\node[label=100:$n$, circle, fill=red] (N) at (4,0) {}; 
					\node[label=80:$o$, square, fill=red] (O) at (6,-2) {}; 
					\node[label=80:$p$, circle, fill=blue] (P) at (6,0) {}; 
					\node[label=$q$, circle, fill=red] (Q) at (8,4) {}; 
					\node[label=100:$r$, square, fill=blue] (R) at (8,2) {}; 
					
					\draw (A)--(B);
					\draw (A)--(C);
					\draw (A)--(D);
					
					\draw (B)--(G);
					\draw (B)--(H);
					
					\draw (C)--(D);
					\draw (C)--(N);
					
					\draw (D)--(E);
					\draw (D)--(F);
					\draw (D)--(G);
					
					\draw (E)--(F);
					\draw (E)--(K);
					\draw (E)--(N);
					\draw (E)--(I);
					
					\draw (F)--(G);
					\draw (F)--(I);
					
					\draw (G)--(H);
					\draw (G)--(I);
					
					\draw (H)--(J);
					
					\draw (I)--(J);
					\draw (I)--(L);
					
					\draw (J)--(L);
					
					\draw (K)--(L);
					\draw (K)--(N);
					\draw (K)--(P);
					\draw (K)--(Q);
					\draw (K)--(R);
					
					\draw (L)--(Q);
					
					\draw (M)--(P);
					\draw (M)--(O);
					\draw (M)--(N);
					
					\draw (N)--(P);
					
					\draw (O)--(P);
					
					\draw (Q)--(R);
					
					
					\end{tikzpicture}
					\caption{Case 1 after the coloring of $G'$ is fixed.}
					\label{fig:case1}
				\end{subfigure}
				~
				\begin{subfigure}[b]{0.45\linewidth}
					\centering
					\begin{tikzpicture}[scale=0.4,
					square/.style = {regular polygon, regular polygon sides=4, scale=4},
					circle/.style = {shape=circle, scale=1.8}]
					\tikzstyle{every node}=[draw, color=black, shape=circle, minimum size=3pt,inner sep=0pt]
					
					\node[label=35:$a$, square, fill=green] (A) at (6,2) {}; 
					\node[label=$b$, circle, fill=red] (B) at (6,4) {}; 
					\node[label=left:$c$, circle, fill=blue] (C) at (8,0) {}; 
					\node[label=100:$d$, circle, fill=red] (D) at (8,2) {}; 
					\node[label=80:$e$, circle, fill=blue] (E) at (10,2) {}; 
					\node[label=$f$, circle, fill=green] (F) at (9,3) {}; 
					\node[label=100:$g$, circle, fill=blue] (G) at (8,4) {}; 
					\node[label=$h$, square, fill=green] (H) at (8,6) {}; 
					\node[label=100:$i$, circle, fill=red] (I) at (10,4) {}; 
					\node[label=$j$, circle, fill=blue] (J) at (10,6) {}; 
					\node[label=100:$k$, circle, fill=red] (K) at (12,2) {}; 
					\node[label=$l$, circle, fill=green] (L) at (12,4) {}; 
					\node[label=100:$m$, circle, fill=red] (M) at (10,-2) {}; 
					\node[label=100:$n$, circle, fill=green] (N) at (10,0) {}; 
					\node[label=80:$o$, square, fill=green] (O) at (12,-2) {}; 
					\node[label=80:$p$, circle, fill=blue] (P) at (12,0) {}; 
					\node[label=$q$, circle, fill=blue] (Q) at (14,4) {}; 
					\node[label=100:$r$, square, fill=green] (R) at (14,2) {}; 
					
					\draw (A)--(B);
					\draw (A)--(C);
					\draw (A)--(D);
					
					\draw (B)--(G);
					\draw (B)--(H);
					
					\draw (C)--(D);
					\draw (C)--(N);
					
					\draw (D)--(E);
					\draw (D)--(F);
					\draw (D)--(G);
					
					\draw (E)--(F);
					\draw (E)--(K);
					\draw (E)--(N);
					\draw (E)--(I);
					
					\draw (F)--(G);
					\draw (F)--(I);
					
					\draw (G)--(H);
					\draw (G)--(I);
					
					\draw (H)--(J);
					
					\draw (I)--(J);
					\draw (I)--(L);
					
					\draw (J)--(L);
					
					\draw (K)--(L);
					\draw (K)--(N);
					\draw (K)--(P);
					\draw (K)--(Q);
					\draw (K)--(R);
					
					\draw (L)--(Q);
					
					\draw (M)--(P);
					\draw (M)--(O);
					\draw (M)--(N);
					
					\draw (N)--(P);
					
					\draw (O)--(P);
					
					\draw (Q)--(R);

					\end{tikzpicture}
					\caption{Case 2 after the coloring of $G'$ is fixed.}
					\label{fig:case2}
				\end{subfigure}
				\caption{Exactly two possible colorings of the edge crossing gadget (up to a permutation of colors).}
				\label{fig:twoCases}
			\end{figure}
			
			In both cases, the color of $a$ is the same as the color of $r$, and the color of $h$ is the same as the color of $o$. Therefore, our gadget transfers the color from $a$ to $r$ and the color from $h$ to $o$ correctly. Moreover, in case 1, $a$ and $h$ are assigned different colors while in case 2 they are assigned the same color. Therefore, when $a$ and $h$ are forced to be true or false and not neutral, our gadget corresponds to all possible truth assignments for these vertices. 
			In Figure~\ref{fig:twoCases}, these two unique colorings are demonstrated.
		\end{proof}
		
		\subsubsection{The proof of Theorem~\ref{thm:main}}

	\begin{thm}\apxmark \label{thm:main}
		There is a polynomial-time reduction from the Monotone NAE3SAT problem to the 3-coloring problem for unit disk segment visibility graphs.
	\end{thm}
	\begin{proof}
		Given an instance of the Monotone NAE3SAT problem, we construct a circuit where every boolean variable is a wire placed horizontally in the plane, on a different row.
		The clauses are placed on the bottom-most row. 
		The truth assignment of the variables are transferred to the clauses via vertical wires, one end connected to the horizontal wire, other end connected to the clause (See Figure~\ref{fig:3satCircuit}).
		
		
		We replace the wires in the described construction by a series of line segments such that the set of segments on a wire has a unique 3-coloring. Three main components of our reduction described in detail can be summarized as follows.
		
		(1) \emph{A long edge} shown in Figure~\ref{fig:longEdge} is used to transfer a color from one end to the other (similar to transferring the truth assignment of a variable). The horizontal line segments are placed on two rows in the Euclidan plane. This configuration, no matter how long, always has a unique 3-coloring. Thus, by assigning a color to any endpoint, one automatically decides which color to appear in the clause gadget.
		
		(2) \emph{A clause gadget} shown in Figure~\ref{fig:naeclause} is modeled by three line segments $xx'$, $yy'$, and $zz'$, not allowing three variables to have the same truth assignment. In our case, the transferred colors. 
		One endpoint from each segment, $x', y', z'$ yield a $K_3$.
		The remaining endpoints, $x, y, z$ do not see each other.
		The color of $x$ is transferred using long edges, and since $x'$ is the other endpoint, the color of $x'$ must be different.
		The same rule applies to $y$ and $z$ as well.
		Thus, if all $x, y, z$ have the same color, then this clause gadget cannot be 3-colored.
		
		(3) \emph{An edge crossing gadget} shown in Figure~\ref{fig:crossingSegments} describes a certificate for an edge crossing in the circuit so that it can be realized as a set of non-intersecting line segments.		
		
		The truth assignments of the variables are determined by a pair of colors.
		In this case, every square endpoint on a wire should receive one of these two colors.
		We guarantee this by setting up the gadget shown in Figure~\ref{fig:initial}.
		The color of the triangle vertex is the ``neutral'' color, which means that the remaining two colors represents ``true'' and ``false'' for the variables.
		Thus, by picking a color for the triangle, we enforce every variable to transfer either true or false.
		Important part is that a variable connects to a clause only by one of the squares.
		Therefore, the truth assignment of a variable is transferred to the clause by the color of a square endpoint.
		
		Given a Monotone NAE3SAT formula with $m$ clauses $C_1, \dots, C_m$ and $n$ variables $q_1, \dots, q_n$, we construct the corresponding unit disk segment visibility graph $G$ as follows:	
		\begin{itemize}
			\item For each variable $q_i$, add a vertex $v_i$ to $G$ together with a long horizontal edge $H_i$ transferring its color.
			\item For each clause $C_i$ and each variable $q_j$ in $C_i$, add a triangle $T_i$ to $G$ together with a long vertical edge $V_j$ transferring the color of $v_j$. 
			\item For each $V_j$ crossing a $H_i$, add an edge crossing gadget (certificate) to $G$ replacing the vertices in the intersection $V_j \cup H_i$.	
		\end{itemize}
		
		Since we have shown that our gadgets interpret a given Monotone NAE3SAT formula and transfer colors correctly, the constructed unit disk segment visibility graph is 3-colorable if and only if the given Monotone NAE3SAT formula has a satisfying assignment. 
		
		\vspace{0.2cm}
		\noindent{\textbf{The time and space complexity.}} For $n$ variables and $m$ clauses, the number of segments on a vertical long ``wire'' is at most $O(m)$, since the colors repeat every three endpoints. 
		Since there are $n$ variables, total number of segments for vertical edges is at most $O(nm)$.
		For every edge crossing, and for every clause, a constant number of line segments is needed.
		In the worst case, there will be $O(nm)$ edge crossings, hence $O(nm)$ line segments for a constant $c$.
		There are $m$ clauses, and thus $O(m)$ edges are needed for the clauses.
		In total, $O(m + nm + nm + m) = O(m + nm)$ segments are enough to model a NAE3SAT instance with $n$ variables and $m$ clauses. 
		
		It is trivial to see that the configurations given in Figures~\ref{fig:longEdge} and~\ref{fig:threeComponentsSeg} take up polynomial space.
		The edge crossing gadget given in Figure~\ref{fig:crossingSegments} has an embedding with polynomially many digits which can be verified by the coordinate system given in Figure~\ref{fig:crossingGadgetSeg}.
		Notice that when a horizontal edge and a vertical edge cross, because of the embedding, two ends of the edge crossing gadget have slightly different $y$-coordinates for the endpoints of the horizontal segments ($a$ and $r$ in Figure~\ref{fig:crossingGadgetSeg}) and slightly different $x$-coordinates for the endpoints of the vertical segments ($h$ and $o$ in Figure~\ref{fig:crossingGadgetSeg}).
		At first, it seems like the total space that the gadget uses will grow with respect to the number of edge crossings.
		However, this is not the case since we can simply use a pair of diagonal segments to shift the position of the upcoming horizontal (resp. vertical) segments back to the initial $y$-coordinate (resp. $x$-coordinate). 
		Because of the repeating color patterns,  the distance between a pair of adjacent variables can be adjusted such that the crossings have enough space to be embedded in.

		As we proved the correctness of our reduction and showed that it is a polynomial-time reduction, the theorem holds. Since the Monotone NAE3SAT problem is NP-complete \cite{Schaefer_complexitySAT}, the 3-coloring problem for unit disk segment visibility graphs is also NP-complete.
	\end{proof}
	
	\begin{proof}[Sketch proof]					
		Three main components of our reduction are as follows.
		
		(1) \emph{A long edge} shown in Figure~\ref{fig:longEdge} is used to transfer a color from one end to the other (similar to transferring the truth assignment of a variable). This configuration, no matter how long, always has a unique 3-coloring (up to permutation). 
		
		(2) \emph{A clause gadget} shown in Figure~\ref{fig:naeclause} is modeled by three line segments $xx'$, $yy'$, and $zz'$, not allowing three variables to have the same truth assignment, in our case, the transferred colors. 
		If all $x, y, z$ have the same color, then this clause gadget cannot be 3-colored.
		
		(3) \emph{An edge crossing gadget} shown in Figure~\ref{fig:crossingSegments} describes a certificate for an edge crossing in the circuit so that it can be realized as a set of non-intersecting line segments.		
		
		Given a Monotone NAE3SAT formula with $m$ clauses $C_1, \dots, C_m$ and $n$ variables $q_1, \dots, q_n$, we construct the corresponding unit disk segment visibility graph $G$ as follows:	
		\begin{itemize}
			\item For each variable $q_i$, add a vertex $v_i$ to $G$ together with a long horizontal edge $H_i$ transferring its color.
			\item For each clause $C_i$ and each variable $q_j$ in $C_i$, add a triangle $T_i$ to $G$ together with a long vertical edge $V_j$ transferring the color of $v_j$. 
			\item For each $V_j$ crossing a $H_i$, add an edge crossing gadget (certificate) to $G$ replacing the vertices in the intersection $V_j \cup H_i$.	
		\end{itemize}			
		
		This reduction correctly maps any instance of the Monotone NAE3SAT problem to some instance of the 3-coloring problem for unit disk segment visibility graphs, in polynomial-time and polynomial-space.
	\end{proof}
	
	Since the Monotone NAE3SAT problem is NP-complete \cite{Schaefer_complexitySAT}, the 3-coloring problem for unit disk segment visibility graphs is also NP-complete by Theorem~\ref{thm:main}.	
	
	\begin{rem}
		Since the 3-coloring problem for unit disk graphs \cite{Graf_udgColoring} is NP-complete, it is also NP-complete for unit disk point visibility graphs by Lemma~\ref{lem:UnitUDVG}. For an alternative reduction to \cite{Graf_udgColoring}, the mentioned gadgets can be utilized with small modifications.
	\end{rem}

		\subsubsection{An example embedding}
		
		We show an example embedding of a single clause $(q \vee s \vee t)$.
		Note that the figures given throughout Section~\ref{sec:3coloringSegment} are to describe the idea behind the proof.
		In an actual embedding, we might need some supplementary segments to transfer the colors properly.
		In Figure~\ref{fig:exampleSegmentEmb}, there are some extra segments around the edge crossing gadgets.
		These segments have no function other than transferring the last seen color.
		This can also be done by altering the length of each segment.
		
		In Figure~\ref{fig:crossingZoom}, we show a zoom-in view of a crossing which appear in the example embedding given in Figure~\ref{fig:exampleSegmentEmb}.
		The shaded area is the edge crossing gadget described in Figure~\ref{fig:crossingSegments}.
		The bold lines denote the segments, and the thin, red lines denote how a color from a long edge is transferred to and from the edge crossing gadget.
		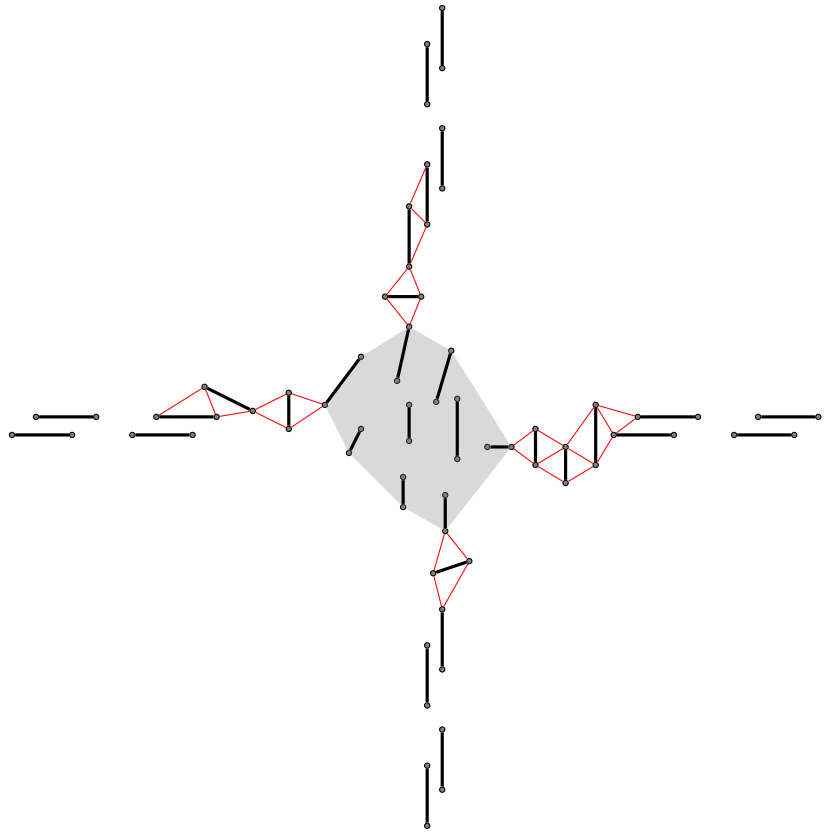
\begin{figure}[htbp]
			\centering
			\begin{tikzpicture}[RED/.style = {fill=red},
			BLUE/.style = {fill=blue},
			GREEN/.style = {fill=green},
			scale=0.8]
			\tikzstyle{every node}=[draw=black, fill=gray, shape=circle, minimum size=2pt,inner sep=0pt];
			\node (A) at (0,0.7) {}; 
			\node (B) at (0.6,1.5) {}; 
			
			\node (C) at (0.4,-0.1) {}; 
			\node (D) at (0.6,0.3) {}; 
			
			\node (E) at (1.4,0.1) {}; 
			\node (F) at (1.4,0.7) {}; 
			
			\node (G) at (1.2,1.1) {}; 
			\node (H) at (1.4,2) {}; 
			\node (I) at (1.85,0.75) {}; 
			\node (J) at (2.1,1.6) {}; 
			
			\node (K) at (2.2,-0.2) {}; 
			\node (L) at (2.2,0.8) {}; 
			
			\node (M) at (1.3,-1) {}; 
			\node (N) at (1.3,-0.5) {}; 
			
			\node (O) at (2,-1.4) {}; 
			\node (P) at (2,-0.8) {}; 
			
			\node (Q) at (2.7,0) {}; 
			\node (R) at (3.1,0) {}; 
			
			\node (L1) at (-2,1) {};
			\node (L2) at (-1.2,0.6) {};
			
			\node (L3) at (-0.6,0.9) {};
			\node (L4) at (-0.6,0.3) {};
			
			\node (R1) at (3.5,0.3) {};
			\node (R2) at (3.5,-0.3) {};
			
			\node (R3) at (4,0) {};
			\node (R4) at (4,-0.6) {};
			
			\node (R5) at (4.5,0.7) {};
			\node (R6) at (4.5,-0.3) {};
			
			\node (U1) at (1,2.5) {};
			\node (U2) at (1.6,2.5) {};
			
			\node (U3) at (1.4,3) {};
			\node (U4) at (1.4,4) {};
			
			\node (D1) at (1.8,-2.1) {};
			\node (D2) at (2.4,-1.9) {};

			\fill[gray,opacity=0.3] (A.center)--(B.center)--(H.center)--(J.center)--(R.center)--(O.center)--(M.center)--(C.center)--(A.center);
			
			\draw[very thick] (A)--(B);
			\draw[very thick] (C)--(D);
			\draw[very thick] (E)--(F);
			\draw[very thick] (G)--(H);
			\draw[very thick] (I)--(J);
			\draw[very thick] (K)--(L);
			\draw[very thick] (M)--(N);
			\draw[very thick] (O)--(P);
			\draw[very thick] (Q)--(R);
			\draw[very thick] (L1)--(L2);
			\draw[very thick] (L3)--(L4);
			\draw[very thick] (R1)--(R2);
			\draw[very thick] (R3)--(R4);
			\draw[very thick] (R5)--(R6);
			\draw[very thick] (D1)--(D2);
			\draw[very thick] (U1)--(U2);
			\draw[very thick] (U3)--(U4);
			
			\begin{scope}[shift={(-4,0.2)}]
			\node (H1) at (0.8,0) {};
			\node (H1') at (1.8,0) {};
			\node (H2) at (1.2,0.3) {};
			\node (H2') at (2.2,0.3) {};
			\end{scope}
			
			\begin{scope}[shift={(-6,0.2)}]
			\node (x) at (0.8,0) {};
			\node (y) at (1.8,0) {};
			\node (z) at (1.2,0.3) {};
			\node (t) at (2.2,0.3) {};
			\draw[very thick] (x)--(y);
			\draw[very thick] (z)--(t);
			\end{scope}
			
			\begin{scope}[shift={(4,0.2)}]
			\node (H3) at (0.8,0) {};
			\node (H3') at (1.8,0) {};
			\node (H4) at (1.2,0.3) {};
			\node (H4') at (2.2,0.3) {};
			\end{scope}
			
			\begin{scope}[shift={(6,0.2)}]
			\node (x) at (0.8,0) {};
			\node (y) at (1.8,0) {};
			\node (z) at (1.2,0.3) {};
			\node (t) at (2.2,0.3) {};
			\draw[very thick] (x)--(y);
			\draw[very thick] (z)--(t);
			\end{scope}
			
			\begin{scope}[shift={(1.7,6)}]
			\node (V1) at (0,-1.3) {};
			\node (V1') at (0,-2.3) {};
			\node (V2) at (0.25,-0.7) {};
			\node (V2') at (0.25,-1.7) {};
			\end{scope}
			
			\begin{scope}[shift={(1.7,8)}]
			\node (x) at (0,-1.3) {};
			\node (y) at (0,-2.3) {};
			\node (z) at (0.25,-0.7) {};
			\node (t) at (0.25,-1.7) {};
			\draw[very thick] (x)--(y);
			\draw[very thick] (z)--(t);
			\end{scope}
			
			\begin{scope}[shift={(1.7,-2)}]
			\node (V3) at (0,-1.3) {};
			\node (V3') at (0,-2.3) {};
			\node (V4) at  (0.25,-0.7){};
			\node (V4') at (0.25,-1.7) {};
			\end{scope}
			
			\begin{scope}[shift={(1.7,-4)}]
			\node (x) at (0,-1.3) {};
			\node (y) at (0,-2.3) {};
			\node (z) at (0.25,-0.7) {};
			\node (t) at (0.25,-1.7) {};
			\draw[very thick] (x)--(y);
			\draw[very thick] (z)--(t);
			\end{scope}

			\draw[very thick] (H1)--(H1');
			\draw[very thick] (H2)--(H2');
			\draw[very thick] (H3)--(H3');
			\draw[very thick] (H4)--(H4');
			
			\draw[very thick] (V1)--(V1');
			\draw[very thick] (V2)--(V2');
			\draw[very thick] (V3)--(V3');
			\draw[very thick] (V4)--(V4');
			
			\tikzstyle{every path}=[draw=red];
			
			\draw (L2)--(H2')--(L1)--(H2);
			\draw (L2)--(L3)--(A)--(L4)--(L2);
			\draw (V1)--(U4)--(V1')--(U3);
			\draw (H)--(U1)--(U3)--(U2)--(H);
			
			\draw (H3)--(H4)--(R5)--(H3)--(R6);
			\draw (R1)--(R3)--(R2)--(R4)--(R6)--(R3)--(R5);
			\draw (R1)--(R)--(R2);
			
			\draw (O)--(D1)--(V4)--(D2)--(O);
			
			\end{tikzpicture}
			\caption{An embedding of an edge crossing gadget zoomed in.}
			\label{fig:crossingZoom}
		\end{figure}

		\begin{figure}
			\hspace{-6em}
			\begin{tikzpicture} [scale=0.42,
			RED/.style = {fill=red},
			BLUE/.style = {fill=blue},
			GREEN/.style = {fill=green},
			]
			
			\tikzstyle{every node}=[draw=black, fill=gray, shape=circle, minimum size=2pt,inner sep=0pt];
			
			\node[GREEN] (A) at (-2,3.1) {};
			\node[BLUE] (B) at (-1,3.1) {};
			\node[RED] (C) at (-1.4,2.7) {};
			\node[GREEN] (D) at (-0.4, 2.7) {};
			\draw[thick] (A)--(B);
			\draw[thick] (C)--(D);
			
			\vedgeL{(0.2,4)}{38}	
			\hedgeT{(0,2.5)}{40}
			
			\vedgeCT{(9.7,2.2)}{8}
			\vedgeCF{(21.7,-21.8)}{8}
			\vedgeCF{(33.7,-33.8)}{10}
			
			\vedgeT{(9.7,-11.7)}{6}
			\vedgeT{(9.7,-23.7)}{6}
			\vedgeT{(9.7,-35.7)}{8}
			\vedgeL{(22,-35.7)}{6}
			
			\hedgeT{(0,-9.5)}{6}
			\hedgeT{(12,-9.5)}{28}
			\hedgeF{(0,-21.5)}{6}
			\hedgeF{(12,-21.5)}{28}
			\hedgeF{(0,-33.5)}{6}
			\hedgeF{(12,-33.5)}{6}
			\hedgeF{(24,-33.5)}{16}
			\crossing{(8,-9.7)}{1}
			\crossing{(8,-21.7)}{1}
			\crossing{(8,-33.7)}{1}
			\crossing{(20,-33.7)}{1}
			
			\hedgeL{(8.6,-44.8)}{10}
			\hedgeF{(24.3,-44.8)}{10}
			
			\nae{(22,-45)}
			
			\end{tikzpicture}
			\caption{An example embedding and coloring of a single clause $(q \vee s \vee t)$ using line segments.}
			\label{fig:exampleSegmentEmb}
		\end{figure}
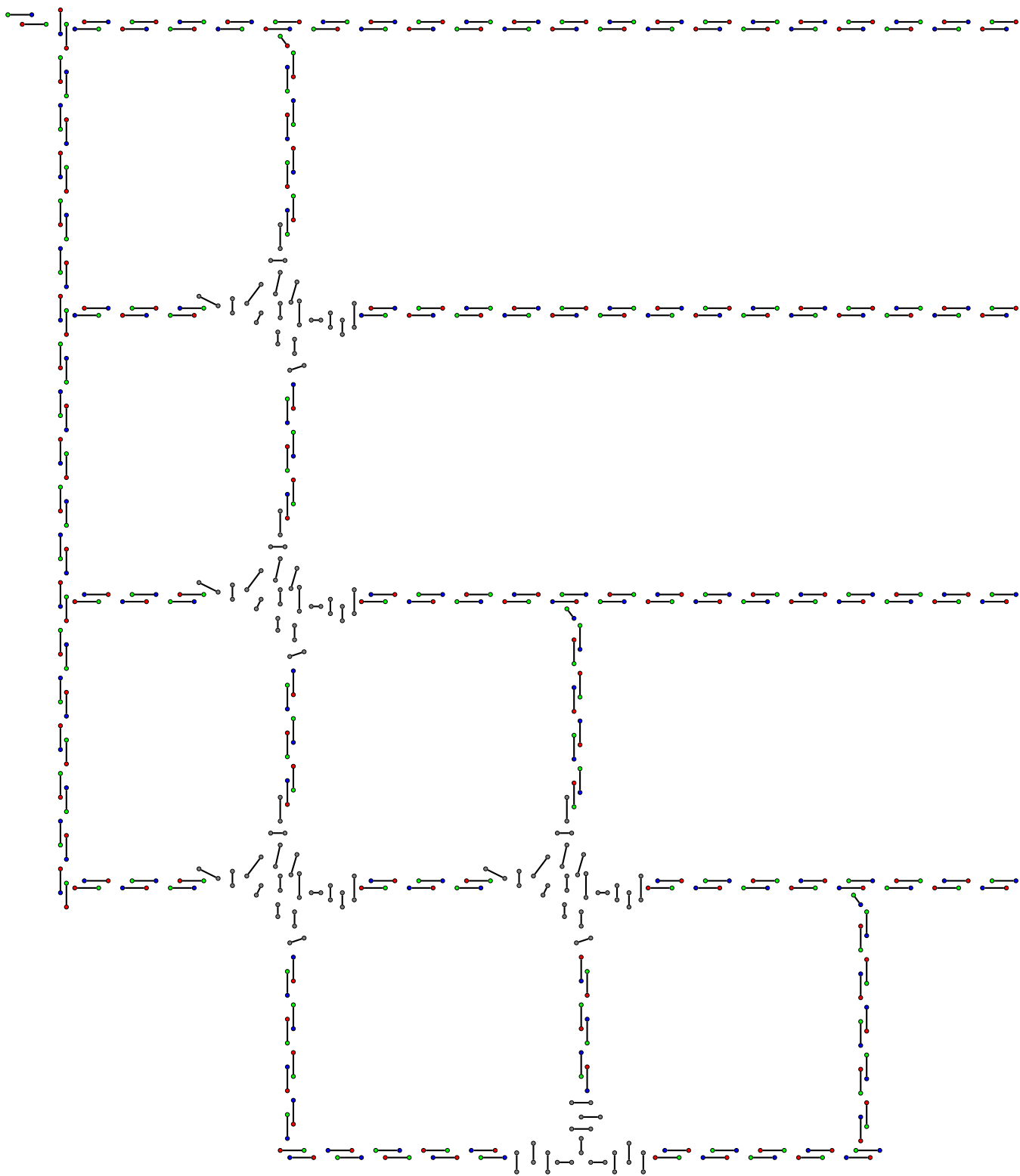
		
		\subsection{The 3-coloring problem for unit disk point visibility graphs}
		
		The idea behind the gadgets used to prove NP-hardness of the 3-coloring problem for unit disk segment visibility graphs can be utilized in order to informally prove the NP-hardness of unit disk point visibility graphs as well, with some small modifications. Having this stated, we give the following remark.
		
		\begin{rem}
			K\'{a}ra et al. showed that there are exactly five cases when the visibility graph of a set of points is 3-colorable (see Figure 3 in \cite{Kara_pointVisChromatic}).
			However, note that our model considers also the Euclidean distances and thus this particular result does not apply to our case when the set of points are not bounded by a circle of diameter 1.
			
			The NP-completeness reduction for unit disk point visibility graphs are straightforward from Gr\"{a}f et al.'s proof of 3-colorability of unit disk graphs \cite{Graf_udgColoring}.	
		\end{rem}
		
		In Figure \ref{fig:longEdgePt}, we show the gadget to transfer the color on a long edge for a unit disk point visibility graph. In Figure~\ref{fig:crossingGadgetPt}, we give the gadget to replace edge crossings. In Figure~\ref{fig:clauseGadgetPt}, we show the embedding of the points in the Euclidean plane. Then, the same reduction given in Section~\ref{sec:3coloringSegment} for unit disk segment visibility graphs can be utilized to prove the NP-completeness of 3-coloring problem for unit disk point visibility graphs.
		
		\begin{figure}[htbp]
			\centering
			\captionsetup[subfigure]{position=t}
			\begin{subfigure}[b]{0.4\linewidth}
				\centering
				\hspace{-0.5em}
				\begin{tikzpicture}[scale=0.6,
				square/.style = {regular polygon, regular polygon sides=4, fill=gray, scale=2.5}	
				]
				\tikzstyle{every node}=[draw=black, fill=black, shape=circle, minimum size=2pt,inner sep=0pt];

				\coordinate (1) at (0,1.5);
				\coordinate (2) at (1.5,1.9);
				\coordinate (3) at (1.5,1.1);
				\coordinate (4) at (3,1.5);
				\coordinate (5) at (4.5,1.9);
				\coordinate (6) at (4.3,1.1);
				\coordinate (7) at (4.5,0);
				\coordinate (8) at (6,1.5);
				\coordinate (9) at (3.9,-1.5);
				\coordinate (10) at (5.1,-1.5);
				\coordinate (11) at (4.5,-3);

				\foreach \i in {2,...,10}
				{
					\node (\i) at (\i) {};
					\draw[gray] (\i) circle (1cm);
				}
				\draw[gray] (1) circle (1cm);
				\draw[gray] (11) circle (1cm);
				
				\node[square] (1) at (1) {};
				\node[square] (11) at (11) {};
				
				\draw (1)--(2)--(3)--(1);
				\draw (2)--(3)--(4)--(2);
				\draw (4)--(5)--(6)--(4);
				\draw (6)--(7)--(5);
				\draw (5)--(6)--(8)--(5);
				\draw (7)--(9)--(10)--(7);
				\draw (9)--(11)--(10);

				\end{tikzpicture}
				\caption{Long edges for unit disk point visibility graphs.}
				\label{fig:longEdgePt}
			\end{subfigure}
			~
			\begin{subfigure}[b]{0.4\linewidth}
				\centering
				\hspace{-0.5em}
				\begin{tikzpicture}[scale=1.5]
				\draw (0,0) to[grid with coordinates] (2.5,2.5);
				
				\tikzstyle{every node}=[draw, fill=black, shape=circle, minimum size=2pt,inner sep=0pt];
				
				\node[label={\footnotesize $a$}] (A) at (0,1) {}; 
				\node[label={\footnotesize $b$}] (B) at (0.25,1.95) {}; 
				\node[label=135:{\footnotesize $c$}] (C) at (0.6,1) {}; 
				\node[label=135:{\footnotesize $d$}] (D) at (0,0.2) {}; 
				\node[label=20:{\footnotesize $e$}] (E) at (1.2,2.1) {}; 
				\node[label=20:{\footnotesize $f$}] (F) at (1.1,1.7) {}; 
				\node[label=left:{\footnotesize $g$}] (G) at (1.4,1) {}; 
				\node[label=-10:{\footnotesize $h$}] (H) at (1.1,0.4) {}; 
				\node[label=100:{\footnotesize $i$}] (I) at (0.9,0) {}; 
				\node[label=above:{\footnotesize $j$}] (J) at (2.1,1.8) {}; 
				\node[label=0:{\footnotesize $k$}] (K) at (1.8,1) {}; 
				\node[label={\footnotesize $l$}] (L) at (1.9,0) {}; 
				\node[label=0:{\footnotesize $m$}] (M) at (2.4,0.85) {}; 

				\tikzstyle{every path}=[draw, color=red];
				
				\draw (A)--(B);
				\draw (A)--(C);
				\draw (A)--(D);
				
				\draw (B)--(E);
				\draw (B)--(F);
				
				\draw (C)--(D);
				\draw (C)--(H);
				\draw (C)--(G);
				\draw (C)--(F);
				
				\draw (D)--(I);
				
				\draw (E)--(F);
				\draw (E)--(J);
				
				\draw (F)--(G);
				\draw (F)--(K);
				
				\draw (G)--(H);
				\draw (G)--(K);
				
				\draw (H)--(I);
				\draw (H)--(L);	
				\draw (H)--(K);		
				
				\draw (I)--(L);
				
				\draw (J)--(K);
				\draw (J)--(M);
				
				\draw (K)--(M);
				
				\draw (L)--(M);
				\end{tikzpicture}
				\caption{The point configuration to be replaced with edge crossings which transfers the color from $a$ to $m$, and from $e$ to $i$.}
				\label{fig:crossingGadgetPt}
			\end{subfigure}
			
			\begin{subfigure}[b]{0.7\linewidth}
				\centering
				\begin{tikzpicture}[scale=0.6, 
				square/.style = {regular polygon, regular polygon sides=4, fill=gray, scale=2.2},
				triangle/.style = {regular polygon, regular polygon sides=3, fill=gray, scale=2.2}
				]
				\tikzstyle{every node}=[draw, fill=black, shape=circle, minimum size=2pt,inner sep=0pt]
				\node[fill=none, shape=rectangle, minimum width=1.6cm, minimum height=0.8cm] at (8.75,2) {};	
				
				\node[square]  (1) at (1.5,3.8) {};
				\node (2) at (1.5,1.9) {};
				\node (3) at (2,2.5) {};
				\node (4) at (3,1.5) {};
				\node (5) at (4.5,1.9) {};
				\node (6) at (4.5,1.1) {};
				\node[square] (7) at (6,1.5) {};
				\node[triangle, rotate=-90] (8) at (7.5,1.5) {};
				\node[triangle, rotate=180] (9) at (8.7,2.5) {};
				\node[triangle, rotate=90] (10) at (10,1.5) {};
				\node[square] (11) at (11.5,1.5) {};
				\node (12) at (13,1.9) {};
				\node (13) at (13,1.1) {};
				\node (14) at (14.5, 1.5) {};
				\node (15) at (15.5, 2.5) {}; 
				\node (16) at (16, 1.9) {};
				\node[square] (17) at (16, 3.8) {};
				\node[square] (18) at (8.7,3.8) {};

				\foreach \i in {1,...,18}
				{
					\draw[gray] (\i) circle (1cm);
				}

				\draw (1)--(2)--(3)--(1);
				\draw (2)--(3)--(4)--(2);
				\draw (4)--(5)--(6)--(4);
				\draw (6)--(7)--(5);
				\draw (7)--(8)--(9)--(10)--(8);
				\draw (10)--(11);
				\draw (11)--(12)--(13)--(11);
				\draw (9)--(18);
				\draw (12)--(13)--(14)--(12);
				\draw (14)--(15)--(16)--(14);
				\draw (15)--(16)--(17)--(15);
				
				\end{tikzpicture}
				\caption{Monotone NAE3SAT clause for a unit disk point visibility graph.}
			\end{subfigure}
			\caption{Three main components to build the NP-hardness gadget for unit disk visibility graph of a set of points.}
			\label{fig:clauseGadgetPt}
		\end{figure}
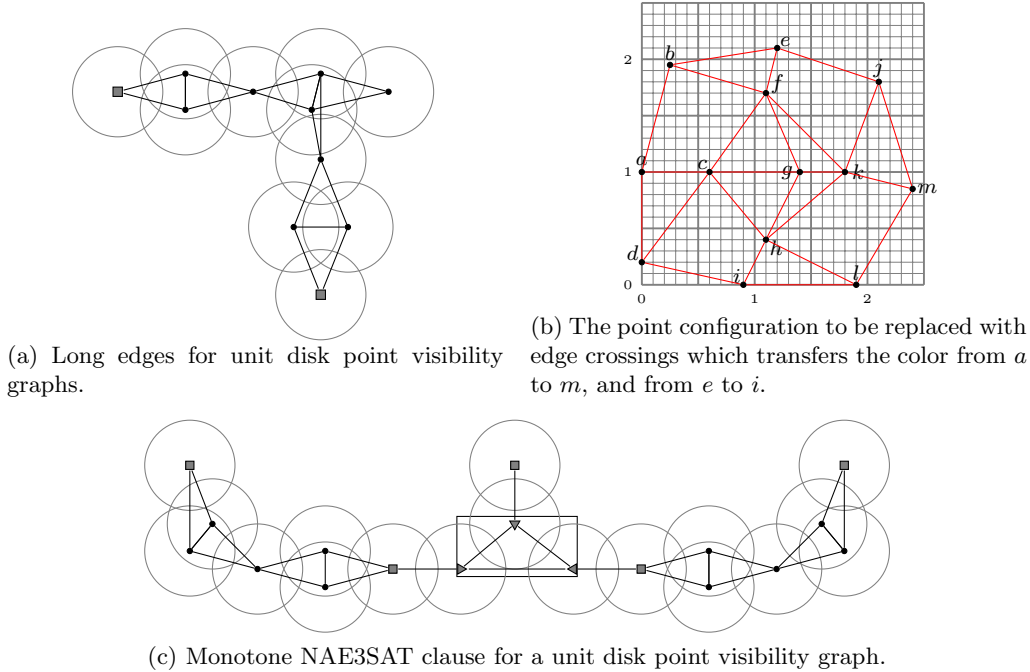
		
		\subsection{The 3-coloring problem for unit disk visibility graphs for polygons with holes}
		
		Before proving Theorem~\ref{thm:withholes}, we describe the gadgets used to construct a polygon with holes (which has a corresponding unit disk visibility graph for the constructed polygon) from a given 4-regular planar graph in more detail, and show that they correctly transform an instance of the 3-coloring problem for 4-regular planar graphs to an instance of the 3-coloring problem for unit disk visibility graphs of polygons with holes.
		
		\subsubsection{The corridors}\label{sec:corridor}
		
		We first describe how we model the edges of a given 4-regular planar graph. In Figure~\ref{fig:polygonEdge}, there are two nodes, $u$ and $v$, and a ``corridor'' which connects them.
		The interior of the polygon is shaded, and the boundaries are indicated with bold lines.
		The visibility edges are indicated using thin lines, and colored red. When the number of vertices on each side of the corridor (excluding $u$ and $v$) is a positive multiple of 3, it is trivial to verify that $u$ and $v$ receive different colors in a proper 3-coloring.
		
		
		Therefore, \emph{a corridor} shown in Figure~\ref{fig:polygonEdge} replaces the edges in a given planar graph.
		We use the same idea which we used to model edges in unit disk segment visibility graphs.
		However, unlike the wires, instead of transferring a color along a long edge, our gadget makes sure that two ends of an edge receives different colors since these ends correspond to adjacent vertices of the given 4-regular planar graph.
		A corridor consists of two polygonal chains $A$ and $B$ with edges $a_1a_2, a_2a_3, \dots, a_{k-1}a_k$ and $b_1b_2, b_2b_3, \dots, b_{k-1}b_k$, respectively. It is trivial to see that we can obtain a unit disk visibility graph for a polygon with holes, where for each $i$, the visibility edges $a_ib_i$, and $b_ia_{i+1}$ exist as visibility edges\footnote{Even if we assume that no three vertices can be collinear, we can slightly perturb the vertices by $\varepsilon$ units where $\varepsilon$ is some positive number which can be represented using polynomially many decimal digits.}.
		This basically describes an induced subgraph with $2k$ vertices, $2k-2$ boundary edges, and $2k-1$ visibility edges.
		Moreover, the largest induced cycle is 3 (which means this is a chordal graph), and each triplet $(a_i,a_{i+1},b_i)$ and $(b_i,b_{i+1},a_i)$ yields a $C_3$.
		Now, suppose that the polygonal chain $A$ has two neighboring vertices $u$ and $v$ where $ua_1$ and $va_k$ are two polygonal edges.
		Assuming that $k = 3c$ for some constant $c \in \mathbb{N}^+$, $u$ and $v$ receive different colors in a 3-coloring of the described subgraph.
		
		\subsubsection{The chambers}\label{sec:chamber}
				
		Now, let us describe the gadget which replaces the vertices in a given planar graph, which we refer to as a ``chamber''. \emph{A chamber} is an induced subgraph with 12 vertices $c_1, \dots, c_{12}$ with boundary edges $c_1c_2$, $c_3c_4$, $c_4c_5$, $c_6c_7$, $c_7c_8$, $c_9c_{10}$, $c_{10}c_{11}$, and $c_{12}c_1$. 
		Figure~\ref{fig:polygonVertex} shows an embedding of a chamber of a polygon $P$ where the interior of the polygon is shaded. The vertex $c_1$ is at the center of some circle $\mathcal{C}$ with radius 1.
		Let us refer to such a vertex as \emph{the central vertex} of the chamber. The vertices $c_i$ for $i=2,3,5,6,8,9,11,12$ are on the boundary of $\mathcal{C}$, and the remaining vertices $c_4, c_7$, and $c_{10}$ are outside $\mathcal{C}$, which means they do not see $c_1$.
		The vertices that are on the boundary of $\mathcal{C}$ are four pair of ``openings'' to the corridors which connect chambers together since the given planar graph is 4-regular. A unit disk is drawn around the central vertex to demonstrate the visibility relations between it and the opening of the corridors. The eight vertices on that unit disk are called the corridor vertices of a chamber, and the remaining three vertices are called the connecting vertices of a chamber. The connecting vertices are essential because the color of the corridor vertices must be dependent only on the central vertex. In this case, if the input graph has two adjacent vertices $u$ and $v$, then there exists a pair of chambers $U$ and $V$, and a corridor with $3c$ vertices which connects $U$ and $V$, and the central vertices $c_u$ and $c_v$ of $U$ and $V$ must receive different colors.
		
		\subsubsection{The proof of Theorem~\ref{thm:withholes}}

	We now show that the 3-coloring problem for unit disk visibility graphs of polygons with holes is NP-hard by giving a reduction from the 3-coloring problem for 4-regular planar graphs \cite{Dailey_4regularplanar}.
	
	\begin{thm}\apxmark  \label{thm:withholes}
		There is a polynomial-time reduction from the 3-coloring problem for 4-regular planar graphs to the 3-coloring problem for unit disk visibility graphs of polygons with holes.
	\end{thm}
	\begin{proof}
		Given a 4-regular planar graph, we construct a polygon with holes. Two main components of our reduction are as follows.
		
		(1) \emph{A corridor} shown in Figure~\ref{fig:polygonEdge} replaces the edges in a given planar graph.
		We use the same idea which we used to model edges in unit disk segment visibility graphs.
		However, instead of transferring a color along a long edge, our gadget makes sure that two ends of an edge receives different colors since the colors of these ends are determined by the colors of the corresponding adjacent vertices of the given 4-regular planar graph. Assuming that $k = 3c$ for some constant $c \in \mathbb{N}^+$, $u$ and $v$ receive different colors in a 3-coloring of the described subgraph.
		
		(2) \emph{A chamber} shown in Figure~\ref{fig:polygonVertex} replaces the vertices in a given planar graph. Since we give a reduction from 4-regular planar graphs, each chamber has exactly four corridors connected to it.
		The big vertex in the center, which is called the central vertex of the chamber, corresponds to a vertex of the given planar graph. In this case, if the input graph has two adjacent vertices $u$ and $v$, then there exists a pair of chambers $U$ and $V$, and a corridor with $3c$ vertices which connects $U$ and $V$, and the central vertices $c_u$ and $c_v$ of $U$ and $V$ must receive different colors.
		
		Given a 4-regular planar graph $H$ on $n$ vertices $v_1, \dots, v_n$, we construct the corresponding polygon $P$ with holes as follows:
		
		\begin{itemize}
			\item For each vertex $v_i$, add a chamber to $P$ whose central vertex is vertex $u_i$.
			\item For each pair of adjacent vertices $(v_i,v_j)$, add a corridor to $P$ between the chambers with central vertices $u_i$ and $u_j$. 
		\end{itemize}
		
		Considering any 3-coloring of $H$, the color given to the vertex $v_i \in H$ can be given to the central vertex $u_i$ of the chamber of $P$ replacing $v_i$, and the colors of central vertices determines the colors of the vertices of corridor, thus a 3-coloring of $P$. Considering any 3-coloring of $P$, the color given to the central vertex $u_i$ of the chamber of $P$ can be given to the vertex $v_i \in H$ replaced by $u_i$. Therefore, $P$ has a 3-coloring if and only if the corresponding color given to the vertices of $H$ yields a 3-coloring.
		
		\vspace{0.2cm}
		\noindent{\textbf{The time and space complexity.}} Given a 4-regular planar graph $H$ on $n$ vertices, we add $n$ chambers to $P$, each having 12 vertices. The positions of the centers of chambers can be determined with respect to any planar embedding of $H$. For each pair of adjacent vertices in $H$, we add a corridor to $P$ between the chambers corresponding to these vertices. The number of vertices on each polygonal chain of a corridor is at most $O(n)$, therefore at most $O(n + n) = O(n)$ vertices in a corridor, in total. Thus, both chambers and corridors take up polynomial space. Since there are $2n$ edges in $H$, there are at most $O(n + 2n^2) = O(n^2)$ vertices in $P$, thus polynomially many edges in $P$. As a result, the given reduction can be done in polynomial time and space.
		
		As we proved the correctness of our reduction and showed that it is a polynomial-time reduction, the theorem holds. Since the 3-coloring problem for 4-regular planar graphs is NP-complete \cite{Dailey_4regularplanar}, the  3-coloring problem for unit disk visibility graphs of polygons with holes is also NP-complete.
	\end{proof}
	
	Since the 3-coloring problem for 4-regular planar graphs is NP-complete \cite{Dailey_4regularplanar}, the 3-coloring problem for unit disk visibility graphs of polygons with holes is also NP-complete by Theorem~\ref{thm:withholes}.

	\def\polyV#1#2{%
		\begin{scope}[shift={#1}, rotate=#2]
			
			\coordinate (1) at (0,0);
			\coordinate (2) at (0.6, -0.8);
			\coordinate (3) at (0.8, -0.6);
			\coordinate (4) at (1.4, 0);
			\coordinate (5) at (0.8, 0.6);
			\coordinate (6) at (0.6, 0.8);
			\coordinate (7) at (0, 1.4);
			\coordinate (8) at (-0.6, 0.8);
			\coordinate (9) at (-0.8, 0.6);
			\coordinate (10) at (-1.4, 0);
			\coordinate (11) at (-0.8, -0.6);
			\coordinate (12) at (-0.6, -0.8);
			
			\draw[thick] (1)--(2);
			\draw[thick] (3)--(4)--(5);
			\draw[thick] (6)--(7)--(8);
			\draw[thick] (9)--(10)--(11);
			\draw[thick] (12)--(1);

			\fill[gray,opacity=0.3] (1) \foreach \i in {2,...,12}
			{
				--(\i)
			};		
			
		\end{scope}
	}
	
	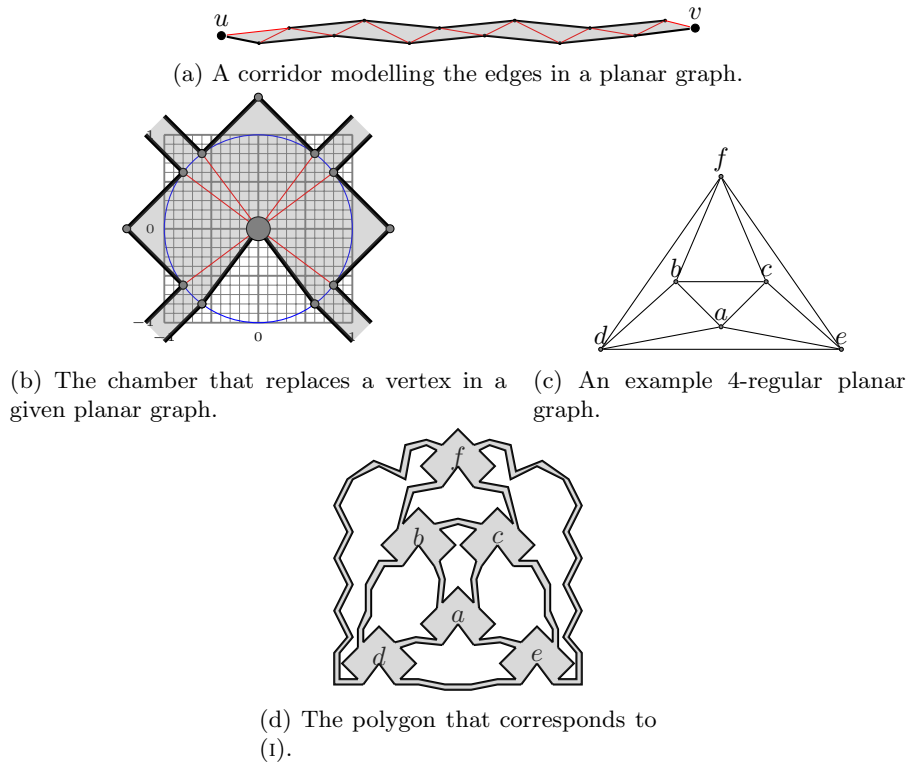
\begin{figure} [htbp]
		\captionsetup[subfigure]{position=b}
		\centering
		
		\begin{subfigure}[b]{0.6\linewidth}
			\centering
			\begin{tikzpicture}
			
			\tikzstyle{every node}=[draw=black, fill=black, shape=circle, minimum size=1pt,inner sep=0pt];

			\foreach [evaluate={\j=int(mod(\i,2));}] \i in {0,...,5}
			{
				\ifthenelse{\j = 1}
				{
					\coordinate (d\i) at (\i,0);
					\coordinate (u\i) at (\i + 0.4 ,0.2);
				}
				{
					\coordinate (d\i) at (\i,-0.1);
					\coordinate (u\i) at (\i + 0.4 ,0.1);
				}		
			}

			\foreach [evaluate={\j=\i+1;}] \i in {0,...,5}
			{
				\ifthenelse{\i = 5}{\draw[red] (u\i)--(d\i);}
				{
					\draw[thick,black] (u\i)--(u\j);
					\draw[thick,black] (d\i)--(d\j);
					\draw[red] (u\i)--(d\i);
					\draw[red] (u\i)--(d\j);				
				}
				\node at (u\i) {};
				\node at (d\i) {};
			}
			
			\fill[gray, opacity=0.3] (-0.5,0)--(u0)--(u1)--(u2)--(u3)--(u4)--(u5)--(5.8,0.1)--(d5)--(d4)--(d3)--(d2)--(d1)--(d0)--(-0.5,0);
			\node[label=$u$,fill=black,scale=3] (u) at (-0.5,0) {};
			\node[label=$v$,fill=black,scale=3] (v) at (5.8,0.1) {};
			\draw[red] (u)--(u0);
			\draw[thick] (u)--(d0);
			\draw[red] (v)--(u5);
			\draw[thick] (v)--(d5);
			\end{tikzpicture}
			\caption{A corridor modelling the edges in a planar graph.}
			\label{fig:polygonEdge}
		\end{subfigure}
		
		\begin{subfigure}[b]{0.4\linewidth}
			\centering
			\begin{tikzpicture} [scale = 1.25]
			\draw (-1,-1) to[grid with coordinates] (1,1);
			\draw[blue] (0,0) circle (1cm);
			
			\begin{scope}[shift = {(-1.4,-0.8)}]
			\tikzstyle{every node}=[draw=black, fill=gray, shape=circle, minimum size=3pt,inner sep=0pt];
			\coordinate (0) at (1.4,0.8);
			\coordinate (1) at (2,0);
			\coordinate (2) at (2.2,0.2);
			\coordinate (3) at (2.8,0.8);
			\coordinate (4) at (2.2,1.4);
			\coordinate (5) at (2,1.6);
			\coordinate (6) at (1.4,2.2);
			\coordinate (7) at (0.8,1.6);
			\coordinate (8) at (0.6,1.4);
			\coordinate (9) at (0,0.8);
			\coordinate (10) at (0.6,0.2);
			\coordinate (11) at (0.8,0);
			
			\coordinate (e1) at (2.4,-0.4);
			\coordinate (e2) at (2.6,-0.2);
			\coordinate (e4) at (2.6,1.8);
			\coordinate (e5) at (2.4,2);
			\coordinate (e7) at (0.4,2);
			\coordinate (e8) at (0.2,1.8);
			\coordinate (e10) at (0.2,-0.2);
			\coordinate (e11) at (0.4,-0.4);	
			
			\draw[ultra thick] (0)--(1)--(e1);
			\draw[ultra thick] (e2)--(2)--(3)--(4)--(e4);
			\draw[ultra thick] (e5)--(5)--(6)--(7)--(e7);
			\draw[ultra thick] (e8)--(8)--(9)--(10)--(e10);
			\draw[ultra thick] (e11)--(11)--(0);
			
			\foreach \i in {2,4,5,7,8,10}
			{
				\draw[red] (0)--(\i);
			}
			
			\foreach \i in {1,...,11}
			{
				\node at (\i) {};
				\ifthenelse { \i = 3 \OR \i = 6 \OR \i = 9}
				{}
				{
					\node at (\i) {};
				}
			}
			
			\node[scale=3] at (0) {};
			
			\fill[gray, opacity=0.3] (0)--(1)--(e1)--(e2)--(2)--(3)--(4)--(e4)--(e5)--(5)--(6)--(7)--(e7)--(e8)--(8)--(9)--(10)--(e10)--(e11)--(11)--(0);

			\end{scope}
			\end{tikzpicture}
			\caption{The chamber that replaces a vertex in a given planar graph.}
			\label{fig:polygonVertex}
		\end{subfigure}
		~
		\begin{subfigure}[b]{0.3\linewidth}
			\centering
			\begin{tikzpicture}
			\tikzstyle{every node}=[draw=black, fill=gray, shape=circle, minimum size=1.5pt,inner sep=0pt];
			\node[label=$a$] (a) at (0,0) {};
			\node[label=$b$] (b) at (-0.6,0.6) {};
			\node[label=$c$] (c) at (0.6,0.6) {};
			\node[label=$d$] (d) at (-1.6,-0.3) {};
			\node[label=$e$] (e) at (1.6,-0.3) {};
			\node[label=$f$] (f) at (0,2) {};
			
			\draw (a)--(b)--(c)--(a);
			\draw (d)--(e)--(f)--(d);
			\draw (a)--(d);
			\draw (a)--(e);
			\draw (b)--(d);
			\draw (b)--(f);
			\draw (c)--(e);
			\draw (c)--(f);
			
			\end{tikzpicture}
			\caption{An example 4-regular planar graph.}
			\label{fig:exampleGraph}
		\end{subfigure}
		~
		\begin{subfigure}[b]{0.32\linewidth}
			\centering
			\begin{tikzpicture}[scale=0.35]
			\node at (0,0.3) {$a$};
			\node at (-1.5,3.3) {$b$};
			\node at (1.5,3.3) {$c$};
			\node at (-3,-1.2) {$d$};
			\node at (3,-1.2) {$e$};
			\node at (0,6.3) {$f$};
			\tikzstyle{every node}=[draw=black, fill=gray, shape=circle, minimum size=1.5pt,inner sep=0pt];
			\polyV{(0,0)}{0};
			\polyV{(-1.5,3)}{0};
			\polyV{(1.5,3)}{0};
			\polyV{(-3,-1.5)}{0};
			\polyV{(3,-1.5)}{0};
			\polyV{(0,6)}{0};
			\tikzstyle{every path}=[thick];
			
			\draw (-0.6,0.8)--(-0.5,1.7)--(-0.7,2.4);
			\draw (-0.8,0.6)--(-0.7,1.7)--(-0.9,2.2);
			\fill[gray,opacity=0.3] (-0.6,0.8)--(-0.5,1.7)--(-0.7,2.4)--(-0.9,2.2)--(-0.7,1.7)--(-0.8,0.6);
			
			\draw (0.6,0.8)--(0.5,1.7)--(0.7,2.4);
			\draw (0.8,0.6)--(0.7,1.7)--(0.9,2.2);
			\fill[gray,opacity=0.3] (0.6,0.8)--(0.5,1.7)--(0.7,2.4)--(0.9,2.2)--(0.7,1.7)--(0.8,0.6);
			
			\draw (-0.8,-0.6)--(-1.6,-0.5)--(-2.4,-0.7);
			\draw (-0.6,-0.8)--(-1.6,-0.7)--(-2.2,-0.9);
			\fill[gray,opacity=0.3] (-0.8,-0.6)--(-1.6,-0.5)--(-2.4,-0.7)--(-2.2,-0.9)--(-1.6,-0.7)--(-0.6,-0.8);
			
			\draw (0.8,-0.6)--(1.6,-0.5)--(2.4,-0.7);
			\draw (0.6,-0.8)--(1.6,-0.7)--(2.2,-0.9);
			\fill[gray,opacity=0.3] (0.8,-0.6)--(1.6,-0.5)--(2.4,-0.7)--(2.2,-0.9)--(1.6,-0.7)--(0.6,-0.8);
			
			\draw (-0.9,3.8)--(0,4)--(0.9,3.8);
			\draw (-0.7,3.6)--(0,3.8)--(0.7,3.6);
			\fill[gray,opacity=0.3] (-0.9,3.8)--(0,4)--(0.9,3.8)--(0.7,3.6)--(0,3.8)--(-0.7,3.6);
			
			\draw (-2.3,2.4)--(-3,2.4)--(-3.5,1.5)--(-3.5,0.5)--(-3.8,-0.2)--(-3.8,-0.9);
			\draw (-2.1,2.2)--(-2.8,2.2)--(-3.3,1.3)--(-3.3,0.3)--(-3.6,-0.2)--(-3.6,-0.7);
			\fill[gray,opacity=0.3] (-2.3,2.4)--(-3,2.4)--(-3.5,1.5)--(-3.5,0.5)--(-3.8,-0.2)--(-3.8,-0.9)--(-3.6,-0.7)--(-3.6,-0.2)--(-3.3,0.3)--(-3.3,1.3)--(-2.8,2.2)--(-2.1,2.2);
			
			\draw (-2.1,3.8)--(-1.7,5)--(-0.6,5.2);
			\draw (-2.3,3.6)--(-1.9,5.2)--(-0.8,5.4);
			\fill[gray,opacity=0.3] (-2.1,3.8)--(-1.7,5)--(-0.6,5.2)--(-0.8,5.4)--(-1.9,5.2)--(-2.3,3.6);
			
			\draw (2.1,3.8)--(1.7,5)--(0.6,5.2);
			\draw (2.3,3.6)--(1.9,5.2)--(0.8,5.4);
			\fill[gray,opacity=0.3] (2.1,3.8)--(1.7,5)--(0.6,5.2)--(0.8,5.4)--(1.9,5.2)--(2.3,3.6);
			
			
			\draw (2.3,2.4)--(3,2.4)--(3.5,1.5)--(3.5,0.5)--(3.8,-0.2)--(3.8,-0.9);
			\draw (2.1,2.2)--(2.8,2.2)--(3.3,1.3)--(3.3,0.3)--(3.6,-0.2)--(3.6,-0.7);
			\fill[gray,opacity=0.3] (2.3,2.4)--(3,2.4)--(3.5,1.5)--(3.5,0.5)--(3.8,-0.2)--(3.8,-0.9)--(3.6,-0.7)--(3.6,-0.2)--(3.3,0.3)--(3.3,1.3)--(2.8,2.2)--(2.1,2.2);
			
			\draw (-2.2,-2.1)--(-1.5,-2.1)--(-0.5,-2.3)--(0.5,-2.3)--(1.5,-2.1)--(2.2,-2.1);
			\draw (-2.4,-2.3)--(-1.5,-2.3)--(-0.5,-2.5)--(0.5,-2.5)--(1.5,-2.3)--(2.4,-2.3);
			\fill[gray,opacity=0.3] (-2.2,-2.1)--(-1.5,-2.1)--(-0.5,-2.3)--(0.5,-2.3)--(1.5,-2.1)--(2.2,-2.1)--(2.4,-2.3)--(1.5,-2.3)--(0.5,-2.5)--(-0.5,-2.5)--(-1.5,-2.3)--(-2.4,-2.3);
			
			\draw (-3.8,-2.1)--(-4.5,-2.1)--(-4.5,-1.1)--(-4,-0.5)--(-4.5,0.5)--(-4,1.5)--(-4.5,2.5)--(-4,3.5)--(-4.5,4.5)--(-4,5.5)--(-3,6)--(-2,5.5)--(-1.8,6.6)--(-1.2,6.8)--(-0.8,6.6);
			\draw (-3.6,-2.3)--(-4.7,-2.3)--(-4.7,-1.1)--(-4.2,-0.5)--(-4.7,0.5)--(-4.2,1.5)--(-4.7,2.5)--(-4.2,3.5)--(-4.7,4.5)--(-4.2,5.5)--(-3.2,6.2)--(-2.2,5.8)--(-2,6.8)--(-1.2,7)--(-0.6,6.8);
			\fill[gray,opacity=0.3] (-3.8,-2.1)--(-4.5,-2.1)--(-4.5,-1.1)--(-4,-0.5)--(-4.5,0.5)--(-4,1.5)--(-4.5,2.5)--(-4,3.5)--(-4.5,4.5)--(-4,5.5)--(-3,6)--(-2,5.5)--(-1.8,6.6)--(-1.2,6.8)--(-0.8,6.6)--(-0.6,6.8)--(-1.2,7)--(-2,6.8)--(-2.2,5.8)--(-3.2,6.2)--(-4.2,5.5)--(-4.7,4.5)--(-4.2,3.5)--(-4.7,2.5)--(-4.2,1.5)--(-4.7,0.5)--(-4.2,-0.5)--(-4.7,-1.1)--(-4.7,-2.3)--(-3.6,-2.3);
			
			\draw (3.8,-2.1)--(4.5,-2.1)--(4.5,-1.1)--(4,-0.5)--(4.5,0.5)--(4,1.5)--(4.5,2.5)--(4,3.5)--(4.5,4.5)--(4,5.5)--(3,6)--(2,5.5)--(1.8,6.6)--(1.2,6.8)--(0.8,6.6);
			\draw (3.6,-2.3)--(4.7,-2.3)--(4.7,-1.1)--(4.2,-0.5)--(4.7,0.5)--(4.2,1.5)--(4.7,2.5)--(4.2,3.5)--(4.7,4.5)--(4.2,5.5)--(3.2,6.2)--(2.2,5.8)--(2,6.8)--(1.2,7)--(0.6,6.8);
			\fill[gray,opacity=0.3] (3.8,-2.1)--(4.5,-2.1)--(4.5,-1.1)--(4,-0.5)--(4.5,0.5)--(4,1.5)--(4.5,2.5)--(4,3.5)--(4.5,4.5)--(4,5.5)--(3,6)--(2,5.5)--(1.8,6.6)--(1.2,6.8)--(0.8,6.6)--(0.6,6.8)--(1.2,7)--(2,6.8)--(2.2,5.8)--(3.2,6.2)--(4.2,5.5)--(4.7,4.5)--(4.2,3.5)--(4.7,2.5)--(4.2,1.5)--(4.7,0.5)--(4.2,-0.5)--(4.7,-1.1)--(4.7,-2.3)--(3.6,-2.3);
			
			\end{tikzpicture}
			\caption{The polygon that corresponds to ({\sc i}).}
			\label{fig:examplePolygon}
		\end{subfigure}
		
		\caption{The gadgets used in the proof of Theorem~\ref{thm:withholes}.}
		\label{fig:PolHol}
	\end{figure}
	
	\section{Conclusion}
	
	We have introduced the unit disk visibility graphs, which models the real-world scenarios more accurately compared to the conventional visibility graphs, and proved the followings:
	
	\begin{itemize}
		\item Visibility graphs are a proper subclass of the unit disk visibility graphs. 
		
		\item Unit disk graphs are a proper subclass of unit disk point visibility graphs while they are neither a subclass nor a superclass of unit disk visibility graphs of a set of line segments, simple polygons and polygons with holes.
		
		\item The 3-coloring problem for unit disk segment visibility graphs is NP-complete.
		
		\item The 3-coloring problem for unit disk visibility graphs of polygons with holes is NP-complete.
	\end{itemize} 
	
	
	In the gadget used to prove NP-completeness of 3-coloring of unit disk segment visibility graphs, all line segments can be exactly one unit long except the edge crossings. Moreover, the rest of the gadget contains line segments either horizontal or vertical (parallel to $x$ or $y$-axis). Considering these facts, we pose these two interesting questions for reader's consideration: 
	
	\begin{open}
		Is the 3-colorability of unit disk visibility graphs of line segments NP-hard when all the segments are exactly 1 unit long?
	\end{open}
	
	\begin{open}
		Is the 3-colorability of unit disk visibility graphs of line segments NP-hard when all the segments are either vertical or horizontal?
	\end{open}
	
	As the above results show that unit disk visibility graphs are not included in the (hierarchic) intersection of unit disk graphs and visibility graphs, we would like to study the following problems which may have interesting results on unit disk visibility graphs.
	
	\begin{open}
		The maximum clique problem on unit disk visibility graphs.
	\end{open}
	
	This problem for unit disk graphs can be solved in polynomial time given \cite{Clark_UDmaxclique} or even without \cite{Raghavan_robust} the representation.
	Since the algorithm described by Clark et al. \cite{Clark_UDmaxclique} does not apply to unit disk visibility graphs due to possible obstacles between disks, it is an interesting problem to study for unit disk visibility graphs.
	
	\begin{open}
		The chromatic number problem on unit disk visibility graphs of polygons.
	\end{open}
	
	In \cite{Cagirici_chromaticpolygon}, it was proven that for visibility graphs of simple polygons, the 4-coloring problem can be solved in polynomial time, and the 5-coloring problem is NP-complete.
	The 3-coloring (even 4-coloring) problem for the unit disk visibility graphs of simple polygons is yet to be solved.
	
	\begin{open}
		The Hamiltonian cycle problem for unit disk segment visibility graphs.
	\end{open}
	
	Hoffman and T\'{o}th showed that every segment visibility graph yields a Hamiltonian cycle \cite{Hoffman_segment}.
	It is clearly not the case for unit disk segment visibility graphs considering two segments with endpoints on $(0,0)$, $(1,0)$, $(0,1)$, $(0,2)$. Thus, it is left as an open question.

	\bibliographystyle{plain}
	\bibliography{bibliography}	
	
\end{document}